\documentclass[A4,12pt]{article}

\usepackage{amsmath,amssymb,amsfonts,amsthm,mathtools} 
\usepackage{bm,bbm}
\usepackage{booktabs}
\usepackage{color}
\usepackage[dvips,letterpaper]{geometry}
\usepackage{epsfig}
\usepackage{fullpage}
\usepackage{here}
\usepackage{graphicx,psfrag,epsf}
\usepackage{indentfirst}
\usepackage{lineno}
\usepackage{natbib}
\usepackage{setspace}
\usepackage{subfigure} 
\usepackage{subfig}
\usepackage{times}
\usepackage{url}
\usepackage{verbatim}
\usepackage{multirow}
\usepackage[table,xcdraw]{xcolor}
\usepackage{adjustbox}
\usepackage{rotating}

\usepackage{pdfpages} 
\usepackage{pdflscape}
\usepackage{scalefnt}

\usepackage[draft,authormarkuptext=name]{changes}
\usepackage{lipsum}

\usepackage{listings}
\usepackage{tikz}
\usetikzlibrary{matrix,shapes,arrows,positioning}

\definecolor{codegreen}{rgb}{0,0.6,0}
\definecolor{codegray}{rgb}{0.5,0.5,0.5}
\definecolor{codepurple}{rgb}{0.58,0,0.82}
\definecolor{backcolour}{rgb}{0.95,0.95,0.92}

\lstdefinestyle{mystyle}{
    backgroundcolor=\color{backcolour},
    commentstyle=\color{codegreen},
    keywordstyle=\color{magenta},
    numberstyle=\tiny\color{codegray},
    stringstyle=\color{codepurple},
    basicstyle=\footnotesize,
    breakatwhitespace=false,
    breaklines=true,
    captionpos=b,
    keepspaces=true,
    numbers=left,
    numbersep=5pt,
    showspaces=false,
    showstringspaces=false,
    showtabs=true,
    tabsize=2}

\input cyracc.def

\newtheorem{rem}{\emph{Remark}}

\newtheorem{theorem}{\emph{Theorem}}

\def \build#1#2#3{\mathrel{\mathop{#1}\limits^{#2}_{#3}}}

\def \build#1#2#3{\mathrel{\mathop{#1}\limits^{#2}_{#3}}}

\title{\textbf{\Large Log-symmetric regression models for correlated errors with an application to mortality data}}

 \author{\normalsize
  {Helton Saulo} \ and \  {Roberto Vila}\\
  {\footnotesize Department of Statistics, Universidade de Bras\'ilia, Brazil}}

\date{}
 
\begin{document}

\maketitle

\begin{abstract}
Log-symmetric regression models are particularly useful when the response variable is continuous, strictly positive and asymmetric. In this paper, we proposed a class of log-symmetric regression models in the context of correlated errors. The proposed models provide a novel alternative to the existing log-symmetric regression models due to its flexibility in accommodating correlation. We discuss some properties, parameter estimation by the conditional maximum likelihood method and goodness of fit of the proposed model. We also provide expressions for the observed Fisher information matrix. A Monte Carlo simulation study is presented to evaluate the performance of the conditional maximum likelihood estimators. 
Finally, a full analysis of a real-world mortality data set is presented to illustrate the proposed approach. 
\end{abstract}

\section{Introduction}\label{sec:1}

Log-symmetric distributions are obtained when a random variable follows the same distribution as its reciprocal, or when the distribution of a logged random variable is symmetric; see \cite{vp:16a}. 
The log-symmetric family of distributions has as special cases the log-normal, log-Student-$t$ and log-power-exponential distributions, among others. Some of its recent applications are in survival analysis, finance and movie industry; see, for example, \cite{vanegasp:16b}, \cite{saulo2017log} and \cite{vslm:18}.

Recently, some works have been published on log-symmetric regression models; see \cite{vp:16a}, \cite{vanegasp:16b,vanegas2017log} and \cite{franciscosilvia2017}. This class of regression models arises when the distribution of the random errors is a member of the log-symmetric family, being particularly useful when the response variable is strictly positive and follows an asymmetric distribution. Moreover, in these models, either the median or skewness of the response variable can be modeled; see \cite{vp:16a}.

A major drawback of using traditional (Gaussian) or log-symmetric regression models arises when the errors are correlated with each other. In this context, the true standard deviation of the estimated regression coefficients may be underestimated by the standard error of the regression coefficients, and the inferential procedures are no longer strictly applicable. Therefore, methods that take into account or remove autocorrelation are necessary. In this scenario, we introduce in this work a class of log-symmetric regression models capable of accommodating correlation, named log-symmetric-autoregressive and moving average (log-symmetric-ARMAX) models. We obtain the conditional maximum likelihood estimators of the proposed model parameters and evaluate their performance by a Monte Carlo simulation study. We also fit the proposed models to a real-world data set for illustrative purpose.

The rest of the paper proceeds as follows. In Section \ref{sec:2}, we describe the log-symmetric distribution and its corresponding regression model. In Section \ref{sec:3}, we introduce the log-symmetric regression model for correlated data. Moreover, we discuss stationary conditions, parameter estimation, Fisher information and residual analysis. In Section~\ref{sec:4}, we carry out a Monte Carlo simulation study to evaluate the behavior of the estimators of the proposed log-symmetric-ARMAX model parameters. In Section~\ref{sec:5}, we apply the proposed models to a real-world mortality data set which is used to study the possible effects of temperature and pollution on mortality in Los Angeles County. Finally, in Section \ref{sec:6}, we discuss some concluding remarks and future research.

\section{Log-symmetric distribution and its regression model}\label{sec:2}

The class of log-symmetric distributions is obtained by taking the exponential of a symmetric random variable; see \cite{vanegasp:16a}. In other words, let $V$ be a continuous random variable following a symmetric distribution with location parameter $\mu\in\mathbb{R}$, scale parameter $\phi>0$ and a density generating kernel $g$, denoted by $V\sim\textrm{S}(\mu,\phi,g)$, and with probability density function (PDF) given by 
$f_{V}(v;\mu,\phi)=\big({\xi_{nc}}/{\sqrt{\phi}}\big)\, g\big({(v-\mu)^2}/{\phi} \big)$, 
where $v\in\mathbb{R}$, $g(u)>0$ for $u>0$ such that 
$\int^{+\infty}_{-\infty} g(z^2)\,\textrm{d}z = 1/\xi_{nc}$
and $\xi_{nc}$ is a normalizing constant; see \cite{fkn:90}. 

Then, the random variable $Y=\exp(V)$ follows a log-symmetric distribution with PDF 
\begin{equation}\label{eq:ft}
f_{Y}(y;\lambda,\phi)=\dfrac{\xi_{nc}}{\sqrt{\phi}\,y}g\left(\dfrac{1}{\phi}\left(\log\left(\dfrac{y}{\lambda}\right)\right)^2\right), \quad y>0,
\end{equation}
where $\lambda=\exp(\mu)>0$ and $\phi>0$ are the scale and shape parameters and they represent, respectively, the median and skewness (or relative dispersion) of the $Y$ distribution. $g$ is a density generating kernel which may be associated with an additional parameter 
$\vartheta$ (or vector $\bm\vartheta$). In this case, we use the notation $Y\sim\textrm{LS}(\lambda,\phi,g)$. Some special log-symmetric distributions are the log-normal, log-power-ex\-po\-nen\-tial, log-Student-$t$ and log-slash, among others; see \cite{cs:88} and \cite{vanegasp:16a}.

A regression model based on \eqref{eq:ft} was studied by \cite{vp:16a,vanegas2017log}, where for a set of $n$ independent random variables, $Y_1,\ldots,Y_n$ say, such that $Y_{i}\sim\textrm{LS}(\lambda_{i},\phi_{i},g)$, $i=1,\ldots,n$, $Y_i$ satisfies the following functional relation
\begin{equation}\label{eq:logsymreg:01}
Y_{i}=\lambda_{i}\,\epsilon_{i}^{\sqrt{\phi_i}},  \quad \epsilon_{i} \sim \textrm{LS}(1, 1, g),
\end{equation}
or in logarithm terms,
\begin{equation}\label{eq:logsymreg:02}
V_i
=
\log(Y_{i})
= 
\mu_{i}
+
\sqrt{\phi_i}
\varepsilon_{i}, 
\quad i = 1, \ldots, n,
\end{equation}
where 
$\mu_i=\log(\lambda_{i})$, $\varepsilon_{i}=\log(\epsilon_{i})$,
$\lambda_i = \Lambda^{-1}(\bm{x}_i^\top\bm \beta)$ and 
$\phi_i = \Lambda^{-1}(\bm{w}^{\top}_{i}\bm{\tau})$, with 
$\bm{\beta} =(\beta_0,\ldots,\beta_{k})^\top$ and 
$\bm{\tau}=(\tau_0,\ldots,{\tau_{l}})^\top$ being vectors of unknown parameters and 
${\bm{x}}^{\top}_{i}= (1,x_{i1},\ldots, x_{ik})^\top$ and 
${\bm{w}}^{\top}_{i} = (1,w_{i1}, \ldots, w_{il})^\top$ are 
the values of $k$ and $l$  covariates associated with the median 
$\lambda_i$ and skewness $\phi_i$, respectively. 
$\Lambda$ is an invertible link function and its inverse 
function is $\Lambda^{-1}$. Note that $\varepsilon_i \sim \textrm{S}(0,1,g)$ 
and $V_{i} \sim \textrm{S}(\mu_{i},\phi_i,g)$.

The log-likelihood function (without the constant) associated with the log-symmetric regression model defined by \eqref{eq:logsymreg:01} and \eqref{eq:logsymreg:02} is given by
\begin{equation}\label{mleequation}
\ell(\bm \theta)
=
-\dfrac{1}{2}\sum_{i=1}^{n}\log(\phi_i) + \sum_{i=1}^{n}\log(g(z^2_i)),
\end{equation}
where $\bm{\theta} = (\bm{\beta},\bm{\zeta})^\top$ and $z_i=(v_i-\mu_i)/ \sqrt{\phi_i}$, for $i=1,\ldots,n$. The maximum likelihood estimate of ${\bm \theta}$ must be obtained numerically with an iterative method for non-linear optimization  problems. For example, by the Broyden-Fletcher-Goldfarb-Shanno quasi-Newton method; see \cite{mjm:00}.


\section{Log-symmetric regression model for correlated data}\label{sec:3}

Let $\{Y_{t}\}$ be random variables and 
${\mathcal{A}}_{t}=\sigma(Y_{t},Y_{t-1},\ldots,)$
be the $\sigma$-field generated by the information up to time $t$. We assume that the conditional distribution of $Y_{t}$ given ${\mathcal{A}}_{t-1}$ follows a log-symmetric distribution, denoted by $Y_{t}|\mathcal{A}_{t-1} \sim \textrm{LS}(\lambda_{t},\phi_t,g)$, with density 
\begin{equation}\label{e1}
f_{Y_t|\mathcal{A}_{t-1}}(y_t;\lambda_{t},\phi_t|\mathcal{A}_{t-1}) 
= 
\dfrac{\xi_{nc}}{\sqrt{\phi_t}\,y_{t}}
g\left(\dfrac{1}{\phi_t}
\left(\log\left(\dfrac{y_{t}}{\lambda_{t}}\right)\right)^2
\right),  \quad y_{t} >0,
\end{equation}
where $\lambda_{t}=\exp(\mu_{t})>0$ and $\phi_{t}>0$ are the corresponding scale and shape parameters, respectively. By using the relation in \eqref{eq:logsymreg:01}, we can write
\begin{equation*}\label{eq:logsymreg:relation}
h(Y_{t})=\lambda_{t}\,\epsilon_{t}^{\sqrt{\phi_t}}
\end{equation*}
and set $h(Y_{t})=\log (Y_{t})$, to obtain
\begin{equation}\label{eq:logsymregarma:01}
h({Y_{t}})= {\mu_{t}}+\sqrt{\phi_t}{\varepsilon_{t}}, \quad t = 1, \ldots, n,
\end{equation}
where $h(Y_{t})|\mathcal{B}_{t-1} \sim \textrm{S}(\mu_{t},\phi_t,g)$, 
$\phi_t = \Lambda^{-1}(\bm{w}^{\top}_{t}\bm{\tau})$ and 
\begin{equation}\label{eq:mut2}
\mu_{t} = {\rm E}[h(Y_t)|\mathcal{B}_{t-1}]  ={\bm x}_t^{\top}\bm{\beta}  + \varrho_{t}, \quad t=1,\dots,n,
\end{equation}
with ${\mathcal{B}}_{t}=\sigma(h(Y_{t}),h(Y_{t-1}),\ldots,)$ 
being the $\sigma$-field generated by the information up to time $t$, and 
$\varrho_{t}$ denoting a dynamic element with ARMA structure, that is,

\begin{equation}\label{eq:taut}
\varrho_{t} 
= 
\sum\limits_{l=1}^p\kappa_l\,\big(h(Y_{t-l}) - {\bm x}_{t-l}^\top\bm \beta\big) + \sum\limits_{j=1}^q\zeta_j\,r_{t-j},
\end{equation}
where $r_t\coloneqq h(Y_{t})-\mu_{t}$ is a martingale difference sequence (MDS), i.e., ${\rm E}|r_t|<\infty$, and ${\rm E}[r_t|\mathcal{B}_{t-1}]=0$, a.s., for all $t$.
This implies that ${\rm E}[r_t]=0$ for all $t$, and $\mathrm{Cov}[r_s,r_t]=0$
(uncorrelatedness of the sequence) for all $t\neq s$.

By adding $h(Y_{t})-\mu_{t}$ to both sides of \eqref{eq:mut2}, we have
\begin{align}\label{eq:taut-1}
h(Y_t)={\bm x}_t^{\top}\bm{\beta}  
+
\sum\limits_{l=1}^p\kappa_l\,\big(h(Y_{t-l}) - {\bm x}_{t-l}^\top\bm \beta\big) + \sum\limits_{j=1}^q\zeta_j\,r_{t-j} 
+
r_t.
\end{align}
In \eqref{eq:mut2}, \eqref{eq:taut} and \eqref{eq:taut-1}, $h$, ${\bm x}_{t}$, $\bm\beta$, $\bm{w}_{t}$ and $\bm{\tau}$ are as in \eqref{eq:logsymreg:02}, $\eta\in\mathbb{R}$, $\bm\kappa=(\kappa_1,\ldots,\kappa_p)^\top\in\mathbb{R}^p$ and $\bm\zeta=(\zeta_1,\ldots,\zeta_q)^\top\in\mathbb{R}^q$. Note that \eqref{eq:mut2} and \eqref{eq:taut} lead to the notation 
log-symmetric-ARMAX($p, q$), as usual in ARMA models.

\subsection{Stationarity conditions}
\begin{theorem}\label{prop-1}
The
marginal mean of $h(Y_t)$ in the \text{log-symmetric-ARMAX}($p, q$) model is given by
\[
{\rm E}[h(Y_t)]
=
{\bm x}_t^{\top}\bm{\beta},
\]
provided that $\Phi(B):\mathbb{R}\to\mathbb{R}$ is an invertible operator 
(the autoregressive polynomial) defined by 
$\Phi(B) = -\sum_{i=0}^{p}\kappa_i B^i$ with $\kappa_0=-1$,  and
$B^i$ is the lag operator, i.e., $B^i y_t = y_{t-i}$.
\end{theorem}
\begin{proof}
Let $\Theta(B) = \sum_{i=0}^{q}\xi_i B^i$ with $\xi_0=1$, be the moving averages polynomial.
Since
$
\Theta(B)\Phi(B)^{-1}=\sum_{i=0}^{\infty}\psi_i B^i 
$
with  $\psi_0=1$, using \eqref{eq:taut-1},
the \text{log-symmetric-ARMAX}($p, q$) model can be rewritten as
\begin{align}\label{rel}
w_t
=
\sum_{l=1}^{p}\kappa_l\, w_{t-l}+ \sum_{j=1}^{q}\zeta_j\, r_{t-j} + r_t
=
\Theta(B)\Phi(B)^{-1} r_t,
\end{align}
where the error $r_t=h(Y_t)-\mu_{t}$ is a \text{MDS}
and $w_t=h(Y_t)-{\bm x}_t^{\top}\bm{\beta}$. 
Then
\[
{\rm E}[h(Y_t)]
=
{\bm x}_t^{\top}\bm{\beta}+{\rm E}[w_t]
\stackrel{\eqref{rel}}{=}
{\bm x}_t^{\top}\bm{\beta}+\Theta(B) \Phi(B)^{-1}{\rm E}[r_t]
=
{\bm x}_t^{\top}\bm{\beta},
\]
whenever the series $\Theta(B)\Phi(B)^{-1}r_t$ converges absolutely.
\end{proof}
%
%
\begin{theorem}\label{theo-var}
Assuming that $\Theta(B)\Phi(B)^{-1}=\sum_{i=0}^{\infty}\psi_i B^i$ and $\Phi(B)$ is invertible,
we have that
the marginal variance of $h(Y_t)$ in the \text{log-symmetric-ARMAX}($p, q$) model is given by
\[
\mathrm{Var}[h(Y_t)]
=
\xi \,
\sum_{i=0}^{\infty}\psi_i^2\, 
\phi_{t-i}^{1/2},
\]
where $\xi> 0$
is a constant not depending on the parameters. The quantity $\xi$ for some distributions
is presented in Table 1 of \cite{franciscosilvia2017}.
\end{theorem}
\begin{proof}
Since ${\rm E}[r_t|\mathcal{B}_{t-1}]=0$, a.s., for all $t$, 
and $\mathrm{Cov}[r_s,r_t]=0$ for all $t\neq s$,
following the notation of  Theorem \ref{prop-1},
we have
\begin{align}\label{first-pr}
\mathrm{Var}[h(Y_t)]
=
\mathrm{Var}[w_t]
=
\mathrm{Var}[\Theta(B)\,\Phi(B)^{-1} r_t]
=
\mathrm{Var}\bigg[\sum_{i=0}^{\infty}\psi_i B^i r_t\bigg]
=
\sum_{i=0}^{\infty}\psi_i^2\, {\rm Var}[r_{t-i}].
\end{align}

On the other hand, the law of total variance states that 
%
\begin{align}\label{l-var}
\mathrm{Var}[r_t]
&=
{\rm E}\big[\mathrm{Var}[r_t|\mathcal{B}_{t-1}]\big]
+
\mathrm{Var}\big[{\rm E}[r_t|\mathcal{B}_{t-1}]\big] 
\nonumber
\\[0,2cm]
&=
{\rm E}\big[\mathrm{Var}[h(Y_t)|\mathcal{B}_{t-1}]\big].
\end{align}
%
Since $\mathrm{Var}[h(Y_t)|\mathcal{B}_{t-1}]=\xi \phi_t^{1/2}$ a.s.,
combining \eqref{first-pr} and \eqref{l-var}, the proof follows.
\end{proof}
\begin{theorem}\label{theo-cov}
The covariance and correlation of $h(Y_t)$ and $h(Y_{t-k})$ in the \text{log-symmetric-ARMAX} ($p, q$) model are given by
\begin{align*}
\mathrm{Cov}[h(Y_t), h(Y_{t-k})]
&= 
\xi\, 
\sum_{i=0}^{\infty} \psi_i \psi_{i-k}\,
\phi_{t-i}^{1/2}, \quad k>0, 
\\[0,2cm]
\mathrm{Corr}[h(Y_t), h(Y_{t-k})]
&= 
{
\sum_{i=0}^{\infty} \psi_i \psi_{i-k}\,
\phi_{t-i}^{1/2}	
\over 
\prod_{j\in\{0,k\}}
\sqrt{ \sum_{i=0}^{\infty} \psi_i^2\, \phi_{t-j-i}^{1/2} }
},
\end{align*}
respectively.
\end{theorem}
\begin{proof}
Since $w_t=h(Y_t)-{\bm x}_t^{\top}\bm{\beta}$ and $\mathrm{Cov}[r_s,r_t]=0$ for all $t\neq s$,
\begin{align*}
\mathrm{Cov}[h(Y_t), h(Y_{t-k})]
=
\mathrm{Cov}[w_t,w_{t-j}]
&\stackrel{\eqref{rel}}{=}
\mathrm{Cov}\big[\Theta(B)\Phi(B)^{-1} r_t, \Theta(B)\Phi(B)^{-1} r_{t-k}\big] 
\\[0,2cm]
&=
\sum_{i=0}^{\infty} \psi_i \psi_{i-k}\, \mathrm{Var}[r_{t-i}].
\end{align*}
Using \eqref{l-var}
the expression on the right side is equal to
$
\sum_{i=0}^{\infty} \psi_i \psi_{i-k}\,
{\rm E}\big[\mathrm{Var}[h(Y_{t-i})|\mathcal{B}_{t-i-1}]\big].
$
Since
$\mathrm{Var}[h(Y_t)|\mathcal{B}_{t-1}]=\xi \phi_t^{1/2}$ a.s.,
the proof follows.
\end{proof}
\begin{rem}
If the parameter $\phi_t=\phi$ is constant,
${\rm Var}[r_t|\mathcal{B}_{t-1}]={\rm Var}[h(Y_t)|\mathcal{B}_{t-1}]=\xi\phi^{1/2}$, 
a.s., for all $t$
(then the \text{MDS} would be a white noise). Then of Theorems \ref{theo-var} and \ref{theo-cov}, the following stationarity conditions follows (see \cite{maiorcysneiros:18})
\[
\mathrm{Var}[h(Y_t)]
=
\xi\phi^{1/2}\,
\sum_{i=0}^{\infty}\psi_i^2,
\quad
\mathrm{Cov}[h(Y_t), h(Y_{t-k})]
= 
\xi\phi^{1/2}\,
\sum_{i=0}^{\infty} \psi_i \psi_{i-k}
\quad \text{and}
\]
\[
\mathrm{Corr}[h(Y_t), h(Y_{t-k})]
=
{
\sum_{i=0}^{\infty} \psi_i \psi_{i-k}
\over 
\sum_{i=0}^{\infty}\psi_i^2	
}, \quad k>0.
\]
\end{rem}

\subsection{Estimation and inference}
The conditional maximum likelihood method can be used to obtain the model parameter estimates based on the first $m$ observations. Consider the parameter vector 
$\bm\theta=(\bm\beta^{\top},\bm\tau^{\top},\bm\kappa^{\top},\bm\zeta^{\top})^{\top}$ and $m=\max\{p,q\}$, for $n> m$. Then, the conditional likelihood function is given by 
\begin{eqnarray*}
L_{m,n}(\bm\theta)
&=&
\prod_{t=m+1}^n f_{\log(Y_t)|\mathcal{B}_{t-1}}(v_t;\mu_{t},{\phi_t}|\mathcal{B}_{t-1}), 
\quad v_{t} \in \mathbb{R},
\end{eqnarray*}
which implies the following conditional log-likelihood function (without the constant)
\begin{eqnarray}\label{loglikearma}
\ell_{m,n}(\bm\theta)
&=& 
-
\dfrac{1}{2}
\sum_{t=m+1}^{n}
\log(\phi_t)
+ 
\sum_{t=m+1}^{n}\log(g(z^2_t)),
\end{eqnarray}
where $z_t=(v_t-\mu_t)/ \sqrt{\phi_t}$, for $t=m+1,\ldots,n$, 
$\phi_t= \Lambda^{-1}(\bm{w}^{\top}_{t}\bm{\tau})$ and
\begin{align}\label{in-der}
\mu_t
=
\sum_{r=0}^{k} \beta_r\,x_{tr}
+
\sum\limits_{l=1}^p\kappa_l\,\Big(v_{t-l} - \sum_{i=0}^{k} \beta_i\, x_{(t-l)i} \Big) 
+
\sum\limits_{j=1}^q\zeta_j\,r_{t-j}.
\end{align}

The conditional maximum likelihood estimates can be obtained by maximizing the expression defined in \eqref{loglikearma} by equating the score vector $\dot{\bm \ell}(\bm{\theta})$, which contains the first derivatives of $\dot{\bm \ell}(\bm{\theta})$, to zero, providing the likelihood equations. Inference for $\bm \theta$ of the log-symmetric-ARMA($p,q$) model can be based on the asymptotic distribution of the conditional maximum likelihood estimator $\widehat{\bm \theta}$. For $n$ sufficiently large and considering usual regularity conditions \citep{eh:78}, the conditional maximum likelihood estimator converges in distribution to a normal distribution
\[
\sqrt{n}\,[\widehat{{\bm \theta}} -{\bm \theta}] \build{\to}{\cal D}{} \textrm{N}_{2+k+l+p+q}(\bm{0}, {\cal J}({\bm \theta})^{-1}),
\]
as $n \to \infty$, where $\build{\to}{\cal D}{}$ means ``convergence in distribution'' and ${\cal J}({\bm \theta})$ is the corresponding expected Fisher information matrix. In this case, we approximate the expected Fisher information matrix by its observed version obtained from the Hessian matrix
\[
\ddot{\bm \ell}(\bm{\theta})
=
\begin{bmatrix}
\displaystyle
{\partial^2 \ell_{0,1}\over \partial \beta_{r}^2}(\bm\theta) 
&\displaystyle
{\partial^2 \ell_{0,1}\over \partial \beta_{r} \partial \tau_s}(\bm\theta)
&\displaystyle
{\partial^2 \ell_{0,1}\over \partial \beta_{r} \partial\kappa_l}(\bm\theta) 
&\displaystyle
{\partial^2 \ell_{0,1}\over \partial \beta_{r} \partial\zeta_j}(\bm\theta)
\\[0,2cm]
\displaystyle
{\partial^2 \ell_{0,1}\over \partial\tau_s \partial\beta_{r}}(\bm\theta) 
&\displaystyle
{\partial^2 \ell_{0,1}\over \partial \tau_s^2}(\bm\theta)
&\displaystyle
{\partial^2 \ell_{0,1}\over \partial \tau_s \partial \kappa_l}(\bm\theta) 
&\displaystyle
{\partial^2 \ell_{0,1}\over \partial \tau_s \partial \zeta_j}(\bm\theta)
\\[0,2cm]
\displaystyle
{\partial^2 \ell_{0,1}\over \partial\kappa_l \partial\beta_{r}}(\bm\theta) 
&\displaystyle
{\partial^2 \ell_{0,1}\over \partial \kappa_l\partial \tau_s}(\bm\theta)
&\displaystyle
{\partial^2 \ell_{0,1}\over \partial \kappa_l^2}(\bm\theta) 
&\displaystyle
{\partial^2 \ell_{0,1}\over \partial \kappa_l \partial \zeta_j}(\bm\theta)
\\[0,2cm]
\displaystyle
{\partial^2 \ell_{0,1}\over \partial\zeta_j \partial\beta_{r}}(\bm\theta) 
&\displaystyle
{\partial^2 \ell_{0,1}\over \partial \zeta_j\partial \tau_s}(\bm\theta)
&\displaystyle
{\partial^2 \ell_{0,1}\over \partial \zeta_j\partial \kappa_l}(\bm\theta) 
&\displaystyle
{\partial^2 \ell_{0,1}\over \partial \zeta_j^2}(\bm\theta)
\end{bmatrix},
\]
where $r=0,\ldots,k$; $s=0,\ldots, l$; $l=1,\ldots,p$ and $j=1,\ldots,q$.
Since the function $\ell_{0,1}(\bm\theta)$ has continuous second partial derivatives at a given point
${\bm\theta}$ in $\mathbb{R}^{4}$, by Schwarz's Theorem follows that the partial differentiations of this function are commutative at that point, that is,
\[
{\partial^2 \ell_{0,1}\over \partial a \partial b}(\bm\theta)
=
{\partial^2 \ell_{0,1}\over \partial b \partial a}(\bm\theta),
\quad \text{for} \ 
a\neq b \ \text{in} \ \{\beta_r,\tau_s,\kappa_l,\zeta_j\}.
\]

It can easily be seen that the first derivatives of $\ell_{0,1}$ are
\[
\begin{array}{llllll}
\displaystyle
{\partial \ell_{0,1}\over \partial a}
(\bm\theta)
=
{1\over g(z^2_t)}\,
{\partial g(z^2_t)\over \partial a},
\ \
a\in\{\beta_r,\kappa_l,\zeta_j\},
\quad 
& 
\displaystyle
{\partial \ell_{0,1}\over \partial \tau_{s}}
(\bm\theta)
=
-{1\over 2\phi_t}\,
{\partial \phi_t\over \partial\tau_s}\,
+
{1\over g(z^2_t)}\,
{\partial g(z^2_t)\over \partial \tau_s},
\end{array}
\]
the second derivatives are
\begin{align*}
{\partial^2 \ell_{0,1}\over \partial a^2}
(\bm\theta)
& =
-
{1\over (g(z^2_t))^2}\,
{\partial g(z^2_t)\over \partial a}
+
{1\over g(z^2_t)}\,
{\partial^2 g(z^2_t)\over \partial a^2},
\quad
a\in\{\beta_r,\kappa_l,\zeta_j\},
\\[0,2cm]
	{\partial^2 \ell_{0,1}\over \partial \tau_{s}^2}
	(\bm\theta)
	&=
	{1\over 2\phi_t^2}\,
	{\partial \phi_t\over \partial\tau_s}\,
	-
	{1\over 2\phi_t}\,
	{\partial^2 \phi_t\over \partial\tau_s^2}\,
	-
	{1\over (g(z^2_t))^2}\,
	{\partial g(z^2_t)\over \partial \tau_s}
	+
	{1\over g(z^2_t)}\,
	{\partial^2 g(z^2_t)\over \partial \tau^2_s},
\end{align*}
and the mixed derivatives are given by
\begin{align*}
{\partial^2 \ell_{0,1}\over \partial \beta_r\partial a}
(\bm\theta)
&=
-{1\over (g(z^2_t))^2}\,
{\partial g(z^2_t) \over \partial\beta_r}\,
{\partial g(z^2_t)\over \partial a}
+
{1\over g(z^2_t)}\,
{\partial^2 g(z^2_t)\over \partial \beta_r \partial a},
\quad
a\in\{\tau_s,\kappa_l,\zeta_j\},
%
%
%
%
\\[0,2cm]
{\partial^2 \ell_{0,1}\over \partial \tau_s\partial b}
(\bm\theta)
&=
-{1\over (g(z^2_t))^2}\,
{\partial g(z^2_t) \over \partial\tau_s}\,
{\partial g(z^2_t)\over \partial b}
+
{1\over g(z^2_t)}\,
{\partial^2 g(z^2_t)\over \partial \tau_s \partial b},
\quad
b\in\{\kappa_l,\zeta_j\},
%
%
\\[0,2cm]
{\partial^2 \ell_{0,1}\over \partial \kappa_l\partial\zeta_{j}}
(\bm\theta)
&=
-{1\over (g(z^2_t))^2}\,
{\partial g(z^2_t) \over \partial\kappa_l}\,
{\partial g(z^2_t)\over \partial \zeta_j}
+
{1\over g(z^2_t)}\,
{\partial^2 g(z^2_t)\over \partial \kappa_l \partial \zeta_j}.
\end{align*}

Let 
\[
\eta_t\coloneqq 
{z_t\over\sqrt{\phi_t}}
= 
{v_t-\mu_t\over\phi_t}.
\] 
The first derivatives of $g$ are
\[
\begin{array}{llllll}
\displaystyle
{\partial g(z^2_t)\over \partial a}
=
-
{2}\, 
\eta_t  
{\partial \mu_t\over \partial a}
{\partial g\over \partial a}(z^2_t),
\quad
a\in\{\beta_r,\kappa_l,\zeta_j\},
&\qquad
\displaystyle
{\partial g(z^2_t)\over \partial \tau_s}
=
-\eta_t^2\, 
{d \phi_t\over d\tau_s}\,
{\partial g\over \partial \tau_s}(z^2_t),
\end{array}
\]
the second derivatives are
\begin{align*}
{\partial^2 g(z^2_t)\over \partial a^2}
&=
2\left(
{1\over\phi_t}
\left( {\partial \mu_t \over \partial a} \right)^2
-
\eta_t
{\partial^2 \mu_t \over \partial a^2}
\right)
{\partial g \over \partial a}(z_t^2)
-
2\eta_t
{\partial \mu_t \over \partial a}
{\partial^2 g \over \partial a^2}(z_t^2),
\quad
a\in\{\beta_r,\kappa_l,\zeta_j\},
\\[0,2cm]
{\partial^2 g(z^2_t)\over \partial \tau^2_s}
&=
\eta_t
\left(
{2\over\phi_t^2}
\left({d\phi_t\over d\tau_s }\right)^2
-
\eta_t
{d^2\phi_t\over d\tau_s^2 }
\right)
{\partial g\over \partial \tau_s}(z^2_t)
-
\eta_t^2
{d\phi_t\over d\tau_s}
{\partial^2 g\over \partial \tau_s^2}(z^2_t),
%
%
\end{align*}
and the mixed derivatives are given by
\begin{align*}
{\partial^2 g(z^2_t)\over \partial \beta_r \partial \tau_s}
&=
\eta_t
\left(\!
{2\over\phi_t}
{\partial\mu_t\over \partial\beta_r}
\, 
{\partial g \over \partial \tau_s}(z_t^2)
-
\eta_t
\,
{\partial^2 g \over \partial \beta_r\partial\tau_s}(z_t^2)
\right)
{d\phi_t\over d \tau_s},
\\[0,2cm]
{\partial^2 g(z^2_t)\over \partial \beta_r \partial a}
&=
\left(
{2\over\phi_t}
{\partial \mu_t\over \partial\beta_r }
{\partial \mu_t\over\partial a }
-
2\eta_t
{\partial^2 \mu_t\over\partial\beta_r\partial a}
\right)
{\partial g\over\partial a}(z_t^2)
-
2\eta_t
{\partial \mu_t\over \partial a}
{\partial^2 g\over \partial\beta_r\partial a}(z_t^2),
\quad a\in\{\kappa_l,\zeta_j\},
%
\\[0,2cm]
{\partial^2 g(z^2_t)\over \partial \tau_s \partial b}
&=
2\eta_t
\left(
{1\over\phi_t}
{d \phi_t\over d\tau_s}
{\partial g\over \partial b}(z_t^2)
-
{\partial^2 g\over \partial \tau_s \partial b}(z_t^2)
\right)
{\partial \mu_t\over \partial b},
\quad 
b\in\{\kappa_l,\zeta_j\},
%
%
\\[0,2cm]
{\partial^2 g(z^2_t)\over \partial \kappa_l \partial \zeta_j}
&=
\left(
{2\over \phi_t}
{\partial \mu_t\over \partial \kappa_l}
{\partial \mu_t\over \partial \zeta_j}
-
2\eta_t
{\partial^2 \mu_t\over \partial \kappa_l \partial \zeta_j}
\right)
{\partial g\over \partial \zeta_j}(z^2_t)
-
2\eta_t
{\partial \mu_t\over \partial \zeta_j}
{\partial^2 g\over \partial \kappa_l \partial \zeta_j}(z^2_t)
\end{align*}
with
\[
\begin{array}{llll}
\displaystyle
{d \phi_t\over d\tau_s}
=
w_{ts} 
\left({\partial \Lambda\over \partial\tau_s}(\phi_t)\right)^{-1},
& \qquad 
\displaystyle
{d^2 \phi_t\over d\tau_s^2}
=
-
w_{ts} 
\left({\partial \Lambda\over \partial\tau_s}(\phi_t)\right)^{-2}	
{\partial^2 \Lambda\over \partial\tau_s^2}(\phi_t).
\end{array}
\]

By \eqref{in-der}, the first derivatives of $\mu_t$ are
\begin{align*}
{\partial \mu_t\over \partial \beta_r}
&=
x_{tr}
-
\sum\limits_{l=1}^p
\kappa_l\,
x_{(t-l)r} 
-
\sum\limits_{j=1}^q\zeta_j\, {\partial \mu_{t-j} \over \partial \beta_r},
\\[0,2cm]
{\partial \mu_t\over \partial \kappa_l}
&=
v_{t-l} 
- 
\sum_{i=0}^{k} \beta_i\,x_{(t-l)i} 
- 
\sum\limits_{j=1}^q\zeta_j\,
{\partial \mu_{t-j}\over \partial \kappa_l},
\\[0,2cm]
{\partial \mu_t\over \partial \zeta_j}
&=
v_{t-j}-u_{t-j}
- 
\sum\limits_{\tilde{j}=1}^q\zeta_{\tilde{j}}\,
{\partial \mu_{t-\tilde{j}}\over \partial \zeta_j},
\end{align*}
the second derivatives are given by
\begin{align*}
{\partial^2 \mu_t\over \partial \beta_r^2}
=
-
\sum\limits_{j=1}^q\zeta_j\, {\partial^2 \mu_{t-j} \over \partial \beta_r^2}
,
\qquad 
{\partial^2 \mu_t\over \partial \kappa_l^2}
=
-
\sum\limits_{j=1}^q\zeta_j\,
{\partial^2 \mu_{t-j}\over \partial \kappa_l^2},
\qquad 
{\partial^2 \mu_t\over \partial \zeta_j^2}
=
-
{\partial \mu_{t-j}\over \partial \zeta_j} 
-
\sum\limits_{\tilde{j}=1}^q \zeta_{\tilde{j}}\,
{\partial^2 \mu_{t-\tilde{j}}\over \partial \zeta_j^2},
\end{align*}
with mixed derivatives
\begin{align*}
{\partial^2 \mu_t\over \partial\beta_r \partial\kappa_l}
=
- x_{(t-l)r}
-
\sum\limits_{j=1}^q\zeta_j\,
{\partial^2 \mu_{t-j}\over \partial \beta_r\partial \kappa_l},
\qquad
{\partial^2 \mu_t\over \partial a \partial\zeta_j}
=
-{\partial \mu_{t-j}\over \partial a}
-
\sum\limits_{\tilde{j}=1}^q\zeta_{\tilde{j}}\,
{\partial^2 \mu_{t-\tilde{j}}\over \partial a\partial \zeta_j},
\quad 
a\in\{\beta_r,\kappa_l\}.
\end{align*}

\subsection{Residual analysis}

We assess goodness of fit and departures from the assumptions of the model by using the quantile residual, which is given by 
\[
 r^\textrm{Q}_t = \Phi^{-1}(\widehat{S}(t_t|\mathcal{B}_{t-1})), \quad t={m+1},\ldots,n,
\]
where $\Phi^{-1}$ is the inverse function of the standard normal cumulative distribution function (CDF) and $\widehat{S}$ the fitted survival function. The quantile residual has a standard normal distribution when the model is correctly specified. Note that this residual is usually applied to generalized additive models for location, scale and shape; see~\cite{ds:96}.

\section{Monte Carlo simulation}\label{sec:4}

A Monte Carlo simulation study is carried out to evaluate the performance of the conditional maximum likelihood estimators for the \text{log-symmetric-ARMAX}($1, 1$) model under the 
log-normal (LogN), log-Student-$t$ (Log$t$) and log-power-exponential (LogPE) cases. The simulation scenario considered the following model
$$
\log(Y_t)= \beta_{0}+\beta_{1}x_{t-1}
+
\kappa_1\,\big(\log(Y_{t-1}) - \beta_{0}-\beta_{1}x_{t-1}\big) + \zeta_1\,r_{t-1} 
+
r_t\, \quad t=2,\ldots,n,
$$
where $n \in \{100, 300, 500\}$, $\phi_{t}=\phi \in \{1.00, 2.00, 3.00\}$ for all $t$, $\beta_{0}=1$, $\beta_{1}=0.7$, $\kappa_{1}=0.6$, $\zeta_{1}=0.3$, $\vartheta=0.5$ (LogPE) and 
$\vartheta=4$ (Log$t$). The conditional maximum likelihood estimation results are presented in Tables~\ref{table:MC:LogN}--\ref{table:MC:LogPE}. In particular, bias and mean squared error (MSE) are reported 
in these tables. Note that the results allow us to conclude that, as the sample size increases, the bias and MSE of all the estimators decrease, as expected. In general, bias and MSE associated with the conditional maximum likelihood estimates of the Log$t$-ARMAX($1, 1$) model, present the lowest values.

\begin{table}[!ht]
\footnotesize
\centering 
 \renewcommand{\arraystretch}{0.8}
 \renewcommand{\tabcolsep}{0.2cm}
\caption{
Empirical bias and MSE (in parentheses) from simulated data for the indicated conditional maximum likelihood estimators of the LogN-ARMAX($1, 1$) model.} \label{table:MC:LogN}
\begin{adjustbox}{max width=\textwidth}
  \begin{tabular}{llrrrrrrrrr}
	\hline\vspace{-0.2cm} \\ 
$n$&                  & \multicolumn{2}{c}{$\phi=0.5$} && \multicolumn{2}{c}{$\phi=1$} && \multicolumn{2}{c}{$\phi=2$}\\\cline{3-4} \cline{6-7}\cline{9-10}\vspace{-0.2cm}\\
&                     &  Bias    & MSE   &&  Bias    &   MSE &&  Bias    &  MSE\\\cline{3-10} \vspace{-0.2cm}\\
100&$\widehat\phi$    &$-$0.0257 &0.0055 &&$-$0.0514 &0.0218 &&$-$0.1028 &0.0873\\
&$\widehat\beta_{0}$  &$-$0.0072 &0.0631 &&$-$0.0102 &0.1261 &&$-$0.0144 &0.2522\\
&$\widehat\beta_{1}$  &  0.0021  &0.0705 && 0.0030   &0.1410 &&   0.0042 &0.2820\\
&$\widehat\kappa_{1}$ &$-$0.0394 &0.0150 &&$-$0.0394 &0.0150 &&$-$0.0394 &0.0150\\
&$\widehat\zeta_{1}$  &  0.0326  &70.0197&& 0.0326   &0.0197 &&   0.0326 &0.0197\\[0.5ex]
300&$\widehat\phi$    &-0.0093  &0.0018 &&-0.0185 &0.0070 &&-0.0371 &0.0281\\
&$\widehat\beta_{0}$  &  0.0063 &0.0224 && 0.0089 &0.0449 && 0.0126 &0.0897\\
&$\widehat\beta_{1}$  & -0.0136 &0.0234 &&-0.0192 &0.0468 &&-0.0272 &0.0937\\
&$\widehat\kappa_{1}$ & -0.0156 &0.0046 &&-0.0156 &0.0046 &&-0.0156 &0.0046\\
&$\widehat\zeta_{1}$  &  0.0118 &0.0066 && 0.0118 &0.0066 && 0.0118 &0.0066\\[0.5ex]
500&$\widehat\phi$    &-0.0055  &0.0010 &&-0.0110 &0.0041 &&-0.0221 &0.0162\\
&$\widehat\beta_{0}$  &  0.0029 &0.0138 && 0.0041 &0.0275 && 0.0059 &0.0550\\
&$\widehat\beta_{1}$  & -0.0106 &0.0129 &&-0.0149 &0.0258 &&-0.0211 &0.0516\\
&$\widehat\kappa_{1}$ & -0.0077 &0.0025 &&-0.0077 &0.0025 &&-0.0077 &0.0025\\
&$\widehat\zeta_{1}$  &  0.0072 &0.0039 && 0.0072 &0.0039 && 0.0072 &0.0039\\
\hline
\end{tabular}
\end{adjustbox}
\end{table}

\begin{table}[!ht]
\footnotesize
\centering 
 \renewcommand{\arraystretch}{0.8}
 \renewcommand{\tabcolsep}{0.2cm}
\caption{
Empirical bias and MSE (in parentheses) from simulated data for the indicated conditional maximum likelihood estimators of the Log$t$-ARMAX($1, 1$) model.} \label{table:MC:Logt}
\begin{adjustbox}{max width=\textwidth}
  \begin{tabular}{llrrrrrrrrr}
	\hline\vspace{-0.2cm} \\ 
$n$&                  & \multicolumn{2}{c}{$\phi=0.5$} && \multicolumn{2}{c}{$\phi=1$} && \multicolumn{2}{c}{$\phi=2$}\\\cline{3-4} \cline{6-7}\cline{9-10}\vspace{-0.2cm}\\
&                     &  Bias      & MSE   &&  Bias      &   MSE  &&  Bias  &  MSE\\\cline{3-10} \vspace{-0.2cm}\\
100&$\widehat\phi$    & $-$0.0162  &0.0091 &&$-$0.0323   &0.0364  &&$-$0.0646 &0.1456 \\
&$\widehat\beta_{0}$  & $-$0.0107  &0.1113 &&$-$0.0152   &0.2226  &&$-$0.0215 &0.4453\\
&$\widehat\beta_{1}$  & $-$0.0032  &0.0981 &&$-$0.0046   &0.1961  &&$-$0.0065 &0.3923\\
&$\widehat\kappa_{1}$ & $-$0.0389  &0.0125 && $-$0.0389  &0.0125  &&$-$0.0389 &0.0125\\
&$\widehat\zeta_{1}$  &  0.0302    &0.0157 && 0.0302     &0.0157  &&   0.0302 &0.0157\\[0.5ex]
300&$\widehat\phi$    & $-$0.0067 &0.0031 &&$-$0.0134 &0.0124 &&$-$0.0269 & 0.0496\\
&$\widehat\beta_{0}$  & $-$0.0022 &0.0378 &&$-$0.0032 &0.0757 &&$-$0.0045 &0.1513\\
&$\widehat\beta_{1}$  & $-$0.0038 &0.0328 &&$-$0.0054 &0.0656 &&$-$0.0076 &0.1312\\
&$\widehat\kappa_{1}$ & $-$0.0129 &0.0033 &&$-$0.0129 &0.0033 &&$-$0.0129 &0.0033\\
&$\widehat\zeta_{1}$  &  0.0101   &0.0047 &&   0.0101 &0.0047 &&   0.0101 &0.0047\\[0.5ex]
500&$\widehat\phi$    & $-$0.0049  &0.0018 &&$-$0.0099 &0.0072 &&$-$0.0198 &0.0289\\
&$\widehat\beta_{0}$  & $-$0.0059  &0.0219 &&$-$0.0083 &0.0437 &&$-$0.0117 &0.0875\\
&$\widehat\beta_{1}$  & $-$0.0040  &0.0175 &&$-$0.0057 &0.0350 &&$-$0.0081 &0.0701\\
&$\widehat\kappa_{1}$ & $-$0.0083  &0.0020 &&$-$0.0083 &0.0020 &&$-$0.0083 &0.0020\\
&$\widehat\zeta_{1}$  &    0.0066  &0.0029 &&   0.0066 &0.0029 &&   0.0066 &0.0029\\
\hline
\end{tabular}
\end{adjustbox}
\end{table}

\begin{table}[!ht]
\footnotesize
\centering 
 \renewcommand{\arraystretch}{0.8}
 \renewcommand{\tabcolsep}{0.2cm}
\caption{
Empirical bias and MSE (in parentheses) from simulated data for the indicated conditional maximum likelihood estimators of the LogPE-ARMAX($1, 1$) model.} \label{table:MC:LogPE}
\begin{adjustbox}{max width=\textwidth}
  \begin{tabular}{llrrrrrrrrr}
	\hline\vspace{-0.2cm} \\ 
$n$&                  & \multicolumn{2}{c}{$\phi=0.5$} && \multicolumn{2}{c}{$\phi=1$} && \multicolumn{2}{c}{$\phi=2$}\\\cline{3-4} \cline{6-7}\cline{9-10}\vspace{-0.2cm}\\
&                     &  Bias   & MSE   &&  Bias    &   MSE &&  Bias    &    MSE\\\cline{3-10} \vspace{-0.2cm}\\
100&$\widehat\phi$    &$-$0.0217&0.0080 &&$-$0.0434 &0.0321 &&$-$0.0868 &0.1286\\
&$\widehat\beta_{0}$  &0.0067   &0.1690 && 0.0095   &0.3380 && 0.0135   &0.6759\\
&$\widehat\beta_{1}$  &$-$0.0193&0.1566 &&$-$0.0273 &0.3132 &&$-$0.0387 &0.6263\\
&$\widehat\kappa_{1}$ &$-$0.0432&0.0148 &&$-$0.0432 &0.0148 &&$-$ 0.0432&0.0148\\
&$\widehat\zeta_{1}$  &0.0348   &0.0187 && 0.0348   &0.0187 && 0.0348   &0.0187\\[0.5ex]
300&$\widehat\phi$    & $-$0.0076 &0.0026&& $-$0.0151 &0.0103&& $-$0.0302 &0.0411\\
&$\widehat\beta_{0}$  &  0.0019   &0.0524&&  0.0027   &0.1047&&  0.0039   &0.2095\\
&$\widehat\beta_{1}$  &  0.0049   &0.0477&&  0.0069   &0.0955&&  0.0097   &0.1910\\
&$\widehat\kappa_{1}$ & $-$0.0135 &0.0036&& $-$0.0135 &0.0036&& $-$0.0135 &0.0036\\
&$\widehat\zeta_{1}$  &  0.0118   &0.0056&&  0.0118   &0.0056&&  0.0118   &0.0056\\[0.5ex]
500&$\widehat\phi$    & $-$0.0057 &0.0016 &&$-$0.0114 &0.0064 &&$-$0.0229 &0.0256\\
&$\widehat\beta_{0}$  &  0.0028   &0.0335 && 0.0040   &0.0670 &&   0.0056 &0.1340\\
&$\widehat\beta_{1}$  & $-$0.0021 &0.0258 &&$-$0.0029 &0.0515 &&$-$0.0041 &0.1031\\
&$\widehat\kappa_{1}$ & $-$0.0080 &0.0021 &&$-$0.0080 &0.0021 &&$-$0.0080 &0.0021\\
&$\widehat\zeta_{1}$  &  0.0070   &0.0031 && 0.0071   &0.0031 &&   0.0071 &0.0031\\
\hline
\end{tabular}
\end{adjustbox}
\end{table}

\section{Illustrative example}\label{sec:5}
The log-symmetric regression and log-symmetric-ARMAX models are now used to analyze a real-world data set, regarding the possible effects of temperature and pollution on weekly mortality in Los
Angeles County over the 10 year period 1970-1979; see \cite{shusto:17}. We have the following variables from this data set: cardiovascular mortality (response), temperature (covariate) and particulate levels (covariate). Figure \ref{fig:corr-mortality} displays scatter-plots with their corresponding correlations for all these variables presented. From this figure, we detect adequate levels of correlation between the response and the covariates, justifying the use of a linear regression model.

\begin{figure}[!ht]
\vspace{-1cm}
\centering
\psfrag{0}[c][c]{\tiny{0}}
\psfrag{5}[c][c]{\tiny{5}}
\psfrag{10}[c][c]{\tiny{10}}
\psfrag{15}[c][c]{\tiny{15}}
\psfrag{50}[c][c]{\tiny{50}}
\psfrag{60}[c][c]{\tiny{60}}
\psfrag{70}[c][c]{\tiny{70}}
\psfrag{80}[c][c]{\tiny{80}}
\psfrag{90}[c][c]{\tiny{90}}
\psfrag{100}[c][c]{\tiny{100}}
\psfrag{110}[c][c]{\tiny{110}}
\psfrag{120}[c][c]{\tiny{120}}
\psfrag{130}[c][c]{\tiny{130}}
\psfrag{M}[c][c]{\small Mortality}
\psfrag{T}[c][c]{\small Temperature}
\psfrag{P}[c][c]{\small Particulates}
{\includegraphics[height=10cm,width=10cm]{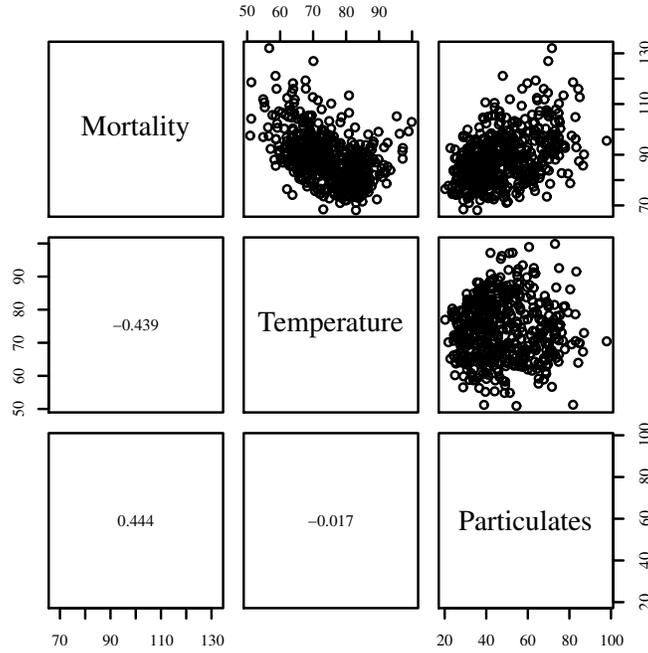}}
\vspace{-1cm}
\caption{Scatterplots and their correlations for the indicated variables with the mortality data.}%
\label{fig:corr-mortality}
\end{figure}

Table~\ref{tab:desc} provides descriptive statistics for the mortality data set, including central tendency statistics, standard
deviation (SD), coefficients of variation (CV), skewness (CS) and kurtosis (CK). From this table, note the presence of skewness and kurtosis in the data distribution; see Figure~\ref{fig:morfigs01}(centre). Note also the presence of autocorrelation; see Figure~\ref{fig:morfigs01}(right).

\begin{table}[!ht]
\footnotesize
\centering
\caption{\small Summary statistics for the mortality data.}\label{tab:desc}
\begin{tabular}{cccccccccc}
\hline
$n$  & Minimum    & Median & Mean     &Maximum      &SD     &CV        & CS      &CK\\
\hline
508  & 68.11      & 87.33  & 88.699   & 132.04      & 9.999 & 11.273\% & 0.804   & 0.981\\
\hline
\end{tabular}
\end{table}

\vspace{-0.2cm}
\begin{figure}[htbp]
\centering
\psfrag{0}[c][c]{\tiny{0}}
\psfrag{5}[c][c]{\tiny{5}}
\psfrag{10}[c][c]{\tiny{10}}
\psfrag{15}[c][c]{\tiny{15}}
\psfrag{50}[c][c]{\tiny{50}}
\psfrag{60}[c][c]{\tiny{60}}
\psfrag{70}[c][c]{\tiny{70}}
\psfrag{80}[c][c]{\tiny{80}}
\psfrag{90}[c][c]{\tiny{90}}
\psfrag{100}[c][c]{\tiny{100}}
\psfrag{110}[c][c]{\tiny{110}}
\psfrag{120}[c][c]{\tiny{120}}
\psfrag{130}[c][c]{\tiny{130}}
\psfrag{1970}[c][c]{\tiny{1970}}
\psfrag{1972}[c][c]{\tiny{1972}}
\psfrag{1974}[c][c]{\tiny{1974}}
\psfrag{1976}[c][c]{\tiny{1976}}
\psfrag{1978}[c][c]{\tiny{1978}}
\psfrag{1980}[c][c]{\tiny{1980}}
\psfrag{-0.2}[c][c]{\tiny{$-$0.2}}
\psfrag{0.0}[c][c]{\tiny{0.0}}
\psfrag{0.2}[c][c]{\tiny{0.2}}
\psfrag{0.1}[c][c]{\tiny{0.1}}
\psfrag{0.3}[c][c]{\tiny{0.3}}
\psfrag{0.4}[c][c]{\tiny{0.4}}
\psfrag{0.5}[c][c]{\tiny{0.5}}
\psfrag{0.6}[c][c]{\tiny{0.6}}
\psfrag{0.7}[c][c]{\tiny{0.7}}
\psfrag{0.8}[c][c]{\tiny{0.8}}
\psfrag{1.0}[c][c]{\tiny{1.0}}
\psfrag{20}[c][c]{\tiny{20.0}}
\psfrag{0.00}[c][c]{\tiny{0.00}}
\psfrag{0.01}[c][c]{\tiny{0.01}}
\psfrag{0.02}[c][c]{\tiny{0.02}}
\psfrag{0.03}[c][c]{\tiny{0.03}}
\psfrag{0.04}[c][c]{\tiny{0.04}}
\psfrag{0.05}[c][c]{\tiny{0.05}}
\psfrag{0.10}[c][c]{\tiny{0.10}}
\psfrag{0.15}[c][c]{\tiny{0.15}}
\psfrag{0.20}[c][c]{\tiny{0.20}}
\psfrag{0.25}[c][c]{\tiny{0.25}}
\psfrag{0.30}[c][c]{\tiny{0.30}}
\psfrag{Y}[c][c]{\tiny{Mortality data}}
\psfrag{I}[c][c]{\tiny{Index}}
\psfrag{L}[c][c]{\tiny{Lag}}
\psfrag{A}[c][c]{\tiny{Autocorrelation function}}
\psfrag{p}[c][c]{\tiny{Density}}
\psfrag{d}[c][c]{\tiny{Mortality data}}
{\includegraphics[height=4.0cm,width=4.0cm]{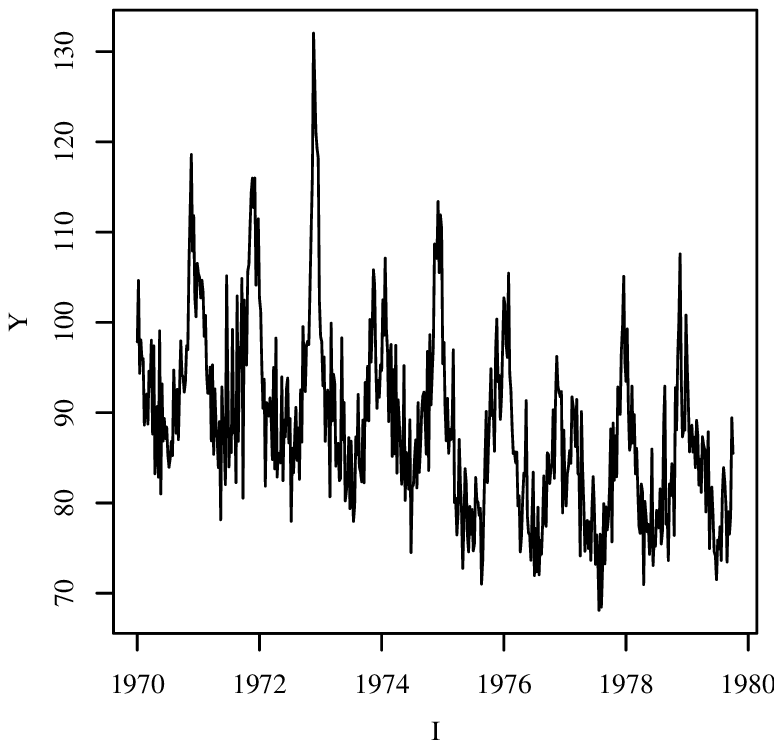}}
{\includegraphics[height=4.0cm,width=4.0cm]{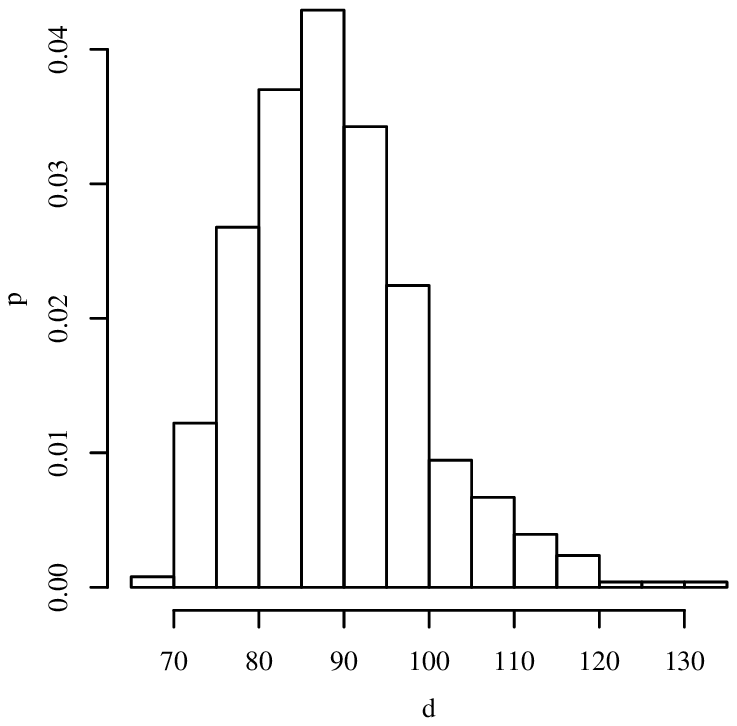}}
{\includegraphics[height=4.0cm,width=4.0cm]{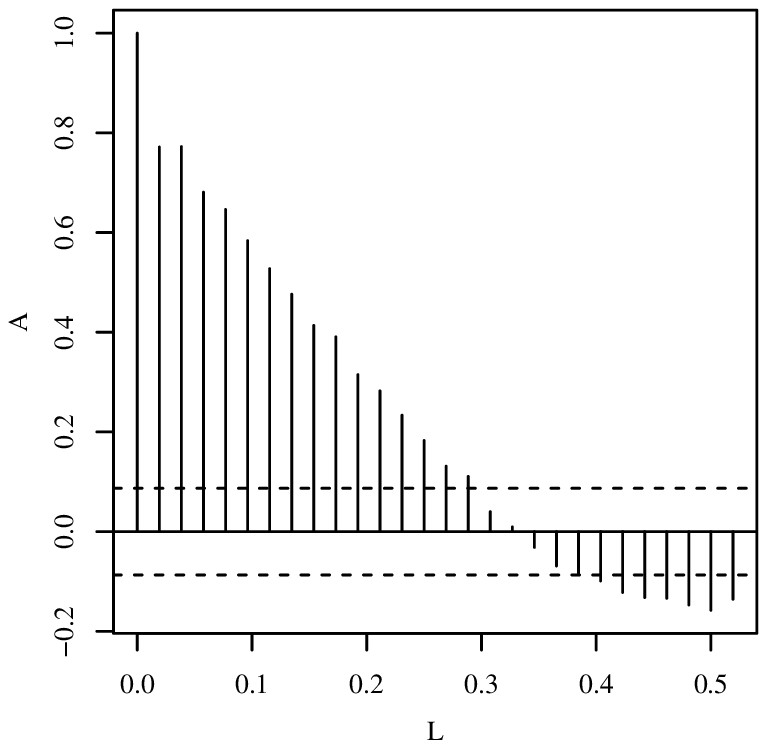}}
\vspace{-0.2cm}
 \caption{\small Timeplot (left), histogram (centre) and autocorrelation (right) function for the mortality data.}
  \label{fig:morfigs01}
\end{figure}

\subsection{Log-symmetric regression results}
We estimate three log-symmetric regression models based on the following special cases: LogN, Log$t$ and LogPE. Based on the scatterplots and timeplot 
shown in Figures \ref{fig:corr-mortality} and \ref{fig:morfigs01}(right) and \cite{shusto:17}, we can set the following final variables: [response] $Y_1$ (mortality) and [covariates] $x_1$ (linear trend), $x_2$ (temperature), $x_3$ (squared temperature) and $X_4$ (particulates). We consider $\phi_{i}=\phi$ for $i=1,\ldots,n$.

Table~\ref{tab:estimates01} reports the estimates, SEs and $p$-values of the $t$-test for the log-symmetric regression model parameters. Furthermore, we report the Akaike (AIC) and Bayesian information (BIC) criteria and the root mean square error (RMSE) to compare the fitted models. From Table~\ref{tab:estimates01}, the three log-symmetric models provide virtually the same adjustments based on the values of RMSE, AIC and BIC. However, the QQ plots with simulated envelope of the quantile residuals for these models show good agreement with the N(0,\,1) distribution only in the 
LogN and Log$t$ regression models; see Figure~\ref{fig:QQRes_ex01}. Nevertheless, the three log-symmetric regression models produce autocorrelated quantile residuals, as shown in Figure~\ref{fig:QQRes_ex01}. Note that the sample autocorrelation and partial autocorrelation functions of the quantile residuals shown in this figure suggest an AR(2) model for the residuals. Thus, a pure log-regression regression model is not adequate and an structure to accommodate correlation is necessary.

\begin{table}[!ht]
\footnotesize
\centering
\renewcommand{\arraystretch}{0.9}
\renewcommand{\tabcolsep}{0.1cm}
\caption{\small {Estimates (with SE in parentheses) and model selection measures for fit to the mortality data.}}\label{tab:estimates01}
\begin{tabular}{llcrcccccccccc}
\hline
Model                &              & Parameter     & \multicolumn{1}{c}{ML estimate}       & $p$-value & RMSE      &   AIC     &   BIC    \\
\hline\\[-0.25cm]
LogN regression model&              & $\beta_{0}$   & 35.4616(2.1630)                       &$<$0.0001  &  0.0692  &$-$1259.696&$-$1234.313     \\
                     &              & $\beta{1}$    &$-$0.0157(0.0011)                      &$<$0.0001  &          &           &              \\
                     &              & $\beta_{2}$   &$-$0.0051(0.0003)                      &$<$0.0001  &          &           &              \\
                     &              & $\beta{3}$    &   0.0002($<$0.0001)                   &$<$0.0001  &          &           &              \\
                     &              & $\beta_{4}$   &   0.0027(0.0002)                      &$<$0.0001  &          &           &              \\                    
                     &              & $\log(\phi)$  & $-$5.3412(0.0627)                     &           &          &           &              \\\hline\\[-0.25cm]
Log$t$ regression model&            & $\beta_{0}$   & 35.4064(2.1630)                       &$<$0.0001  &  0.0692  &$-$1260.442&$-$1235.059     \\
                     &              & $\beta{1}$    &$-$0.0157(0.0011)                      &$<$0.0001  &          &           &              \\
                     &              & $\beta_{2}$   &$-$0.0051(0.0003)                      &$<$0.0001  &          &           &              \\
                     &              & $\beta{3}$    &   0.0002($<$0.0001)                   &$<$0.0001  &          &           &              \\
                     &              & $\beta_{4}$   &   0.0027(0.0002)                      &$<$0.0001  &          &           &              \\                    
                     &              & $\log(\phi)$  & $-$5.5567(0.0725)                     &           &          &           &              \\
                     &              & $\vartheta$   & 9                                     &           &          &           &              \\\hline\\[-0.25cm]
LogPE regression model&             & $\beta_{0}$   & 35.6571(2.1205)                       &$<$0.0001  &  0.0692  &$-$1260.376&$-$1234.994     \\
                     &              & $\beta{1}$    &$-$0.0158(0.0011)                      &$<$0.0001  &          &           &              \\
                     &              & $\beta_{2}$   &$-$0.0051(0.0003)                      &$<$0.0001  &          &           &              \\
                     &              & $\beta{3}$    &   0.0002($<$0.0001)                   &$<$0.0001  &          &           &              \\
                     &              & $\beta_{4}$   &   0.0027(0.0002)                      &$<$0.0001  &          &           &              \\                    
                     &              & $\log(\phi)$  & $-$5.7707(0.0725)                     &           &          &           &              \\
                     &              & $\vartheta$   & 0.24                                  &           &          &           &              \\\hline\\[-0.25cm]
\end{tabular}
\end{table}

\begin{figure}[!ht]
\centering

\psfrag{-1}[c][c]{\tiny{$-$1}}
\psfrag{-2}[c][c]{\tiny{$-$2}}
\psfrag{-3}[c][c]{\tiny{$-$3}}
\psfrag{-4}[c][c]{\tiny{$-$4}}
\psfrag{0}[c][c]{\tiny{0}}
\psfrag{1}[c][c]{\tiny{1}}
\psfrag{2}[c][c]{\tiny{2}}
\psfrag{3}[c][c]{\tiny{3}}
\psfrag{4}[c][c]{\tiny{4}}
\psfrag{5}[c][c]{\tiny{5}}
\psfrag{6}[c][c]{\tiny{6}}
\psfrag{8}[c][c]{\tiny{8}}
\psfrag{-}[c][c]{\tiny{$-$}}
\psfrag{2.0}[c][c]{\tiny{2.0}}
\psfrag{2.5}[c][c]{\tiny{2.5}}
\psfrag{-0.1}[c][c]{\tiny{$-$0.1}}
\psfrag{-2.0}[c][c]{\tiny{$-$2.0}}
\psfrag{-2.5}[c][c]{\tiny{$-$2.5}}
\psfrag{0.0}[c][c]{\tiny{0.0}}
\psfrag{0.1}[c][c]{\tiny{0.1}}
\psfrag{0.2}[c][c]{\tiny{0.2}}
\psfrag{0.3}[c][c]{\tiny{0.3}}
\psfrag{0.4}[c][c]{\tiny{0.4}}
\psfrag{0.5}[c][c]{\tiny{0.5}}
\psfrag{0.6}[c][c]{\tiny{0.6}}
\psfrag{0.7}[c][c]{\tiny{0.7}}
\psfrag{0.8}[c][c]{\tiny{0.8}}
\psfrag{0.9}[c][c]{\tiny{0.9}}
\psfrag{1.0}[c][c]{\tiny{1.0}}
\psfrag{10}[c][c]{\tiny{10.0}}
\psfrag{15}[c][c]{\tiny{15.0}}
\psfrag{20}[c][c]{\tiny{20.0}}
\psfrag{l}[c][c]{\tiny{Lag}}
\psfrag{ac}[c][c]{\tiny{Autocorrelation function}}
\psfrag{pc}[c][c]{\tiny{Partial autocorrelation function}}
\psfrag{EQ}[c][c]{\tiny{Empirical quantile}}
\psfrag{TQ}[c][c]{\tiny{Theoretical quantile}}
\subfigure[LogN regression]{\includegraphics[height=3.8cm,width=3.8cm]{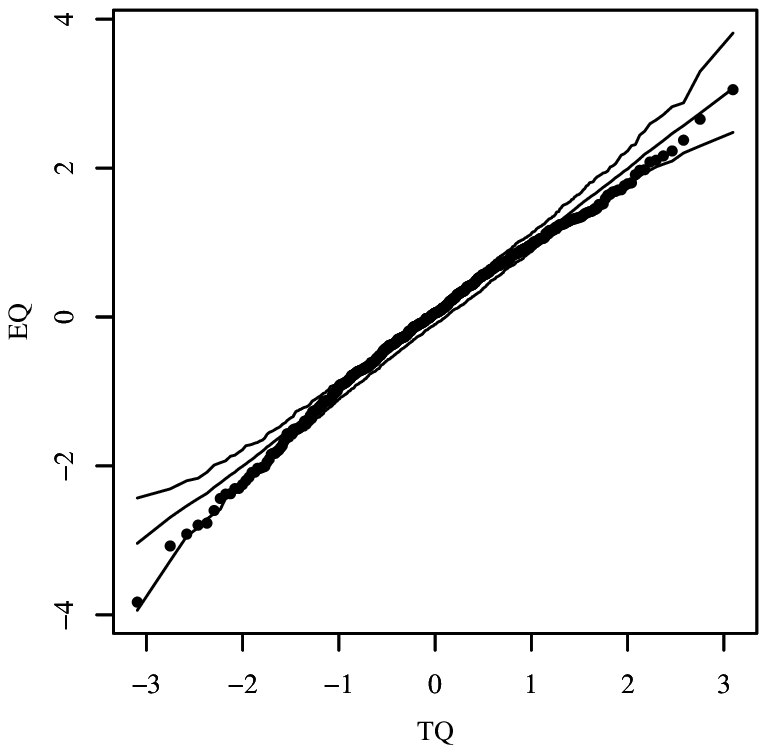}}
\subfigure[LogN regression]{\includegraphics[height=3.8cm,width=3.8cm]{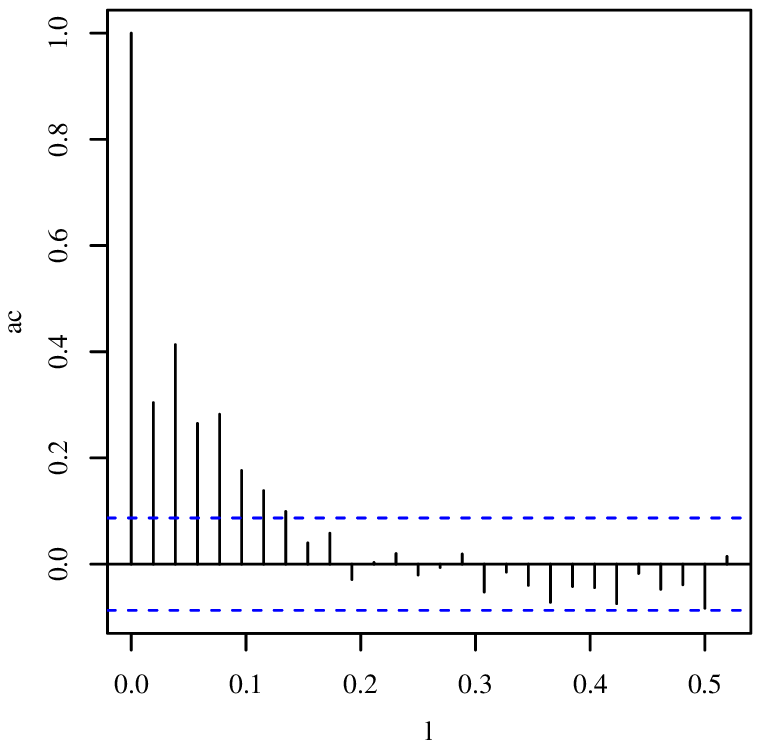}}
\subfigure[LogN regression]{\includegraphics[height=3.8cm,width=3.8cm]{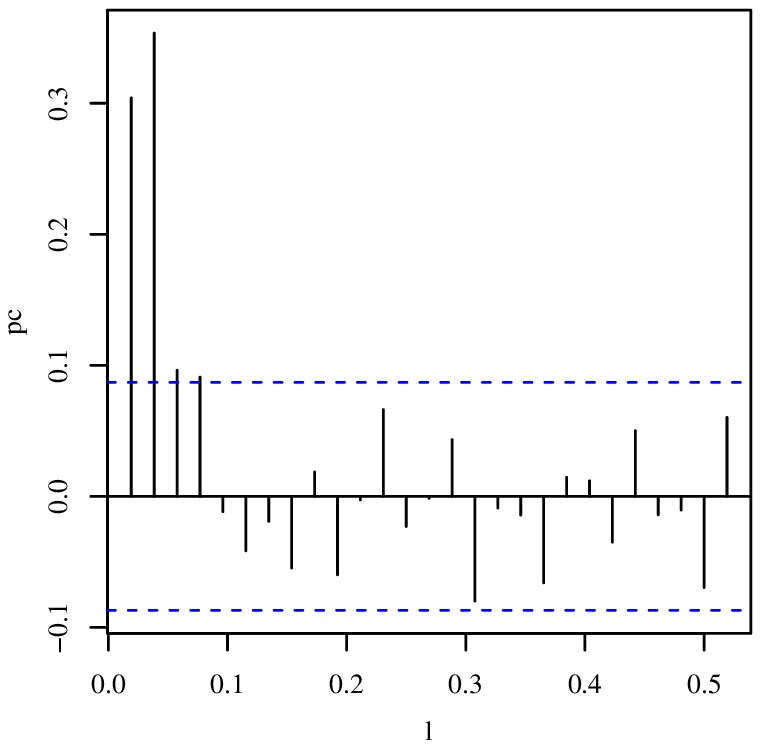}}\\
\subfigure[Log$t$ regression]{\includegraphics[height=3.8cm,width=3.8cm]{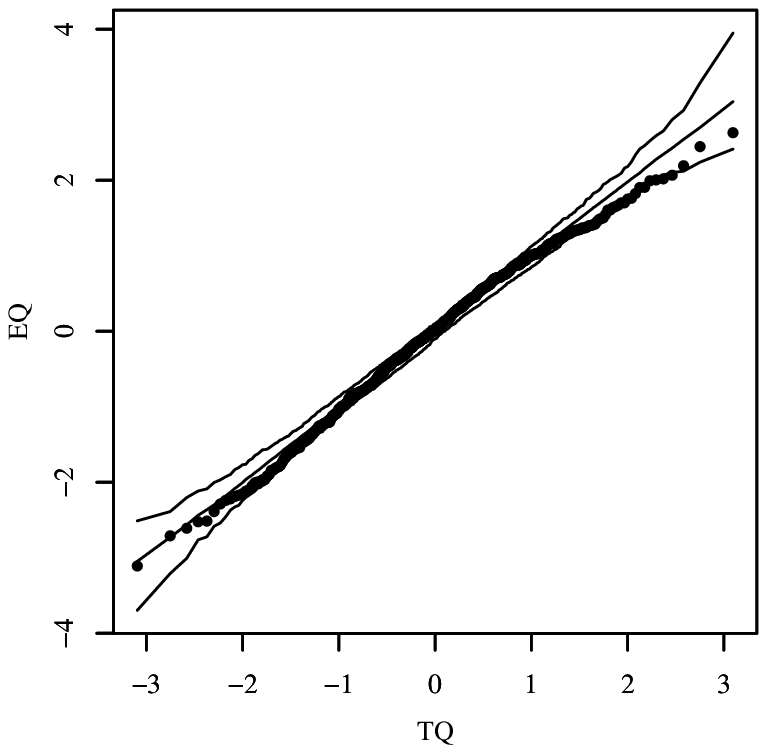}}
\subfigure[Log$t$ regression]{\includegraphics[height=3.8cm,width=3.8cm]{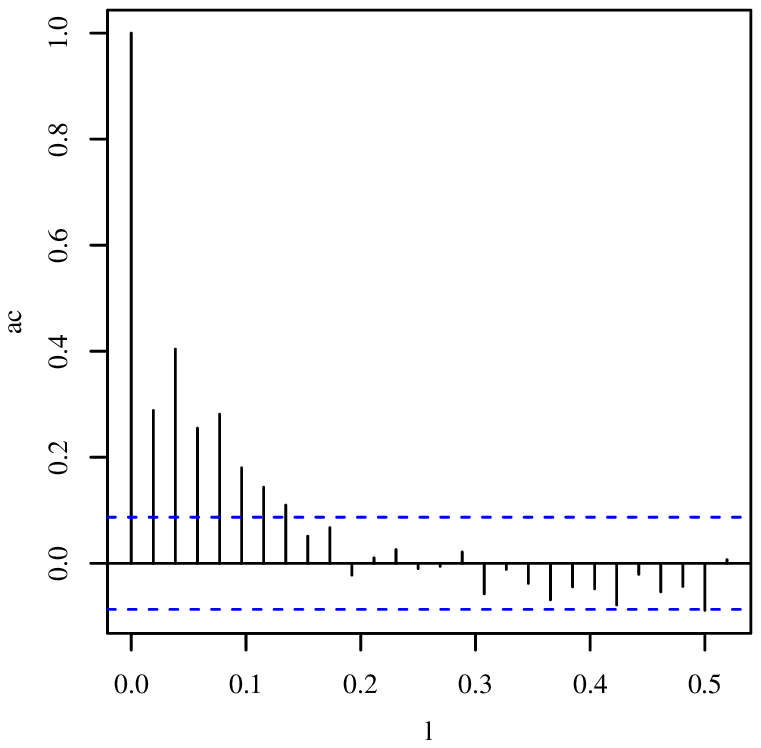}}
\subfigure[Log$t$ regression]{\includegraphics[height=3.8cm,width=3.8cm]{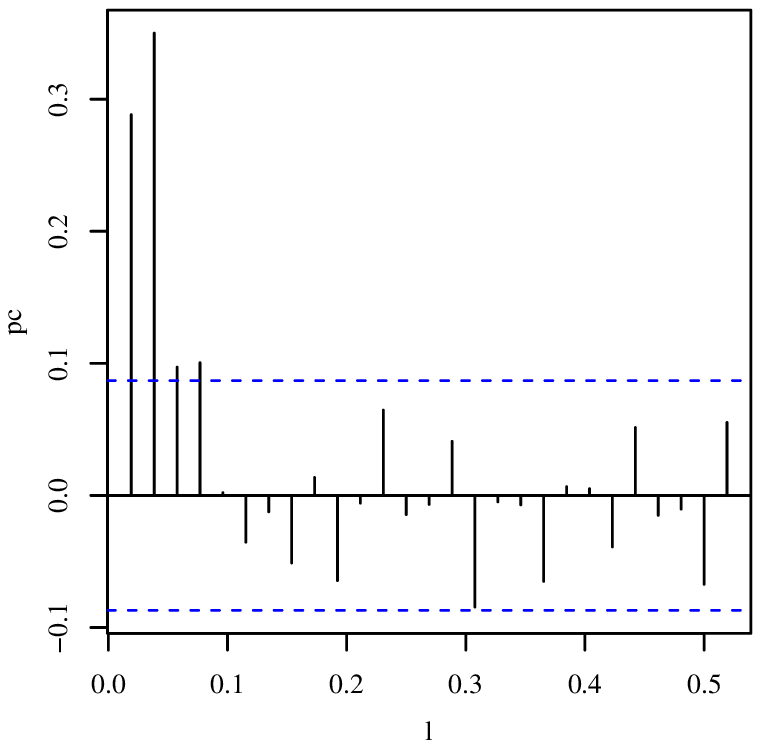}}\\
\subfigure[LogPE regression]{\includegraphics[height=3.8cm,width=3.8cm]{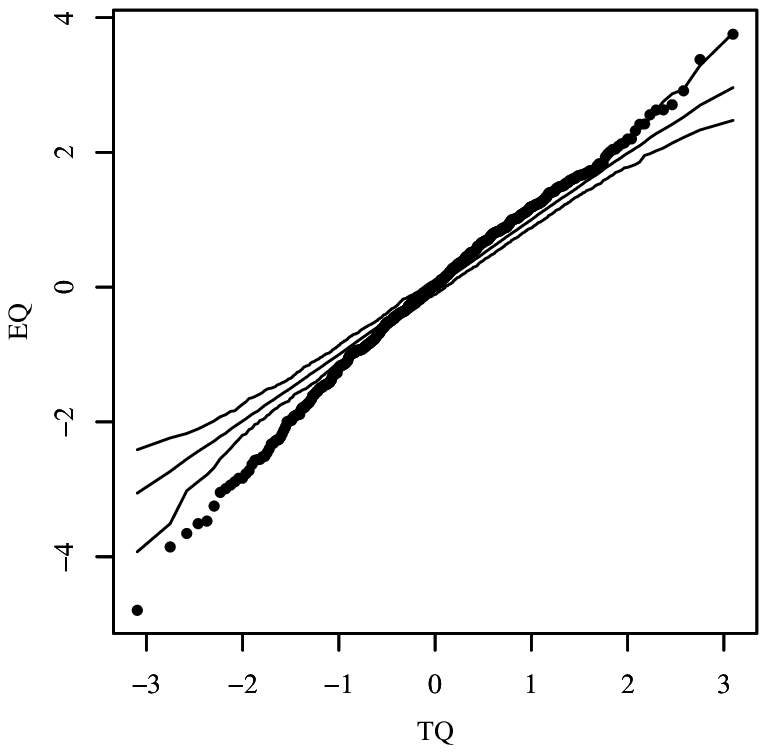}}
\subfigure[LogPE regression]{\includegraphics[height=3.8cm,width=3.8cm]{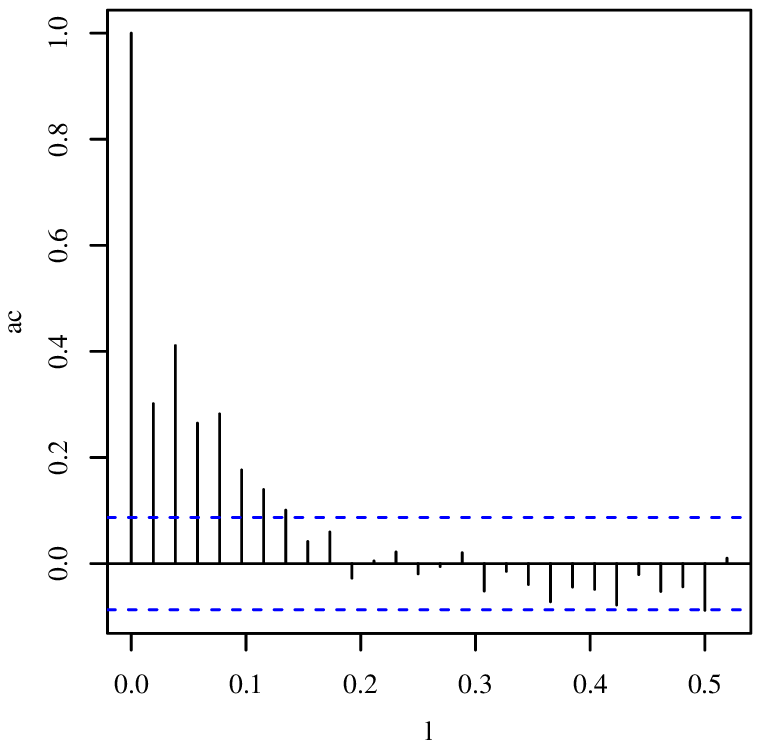}}
\subfigure[LogPE regression]{\includegraphics[height=3.8cm,width=3.8cm]{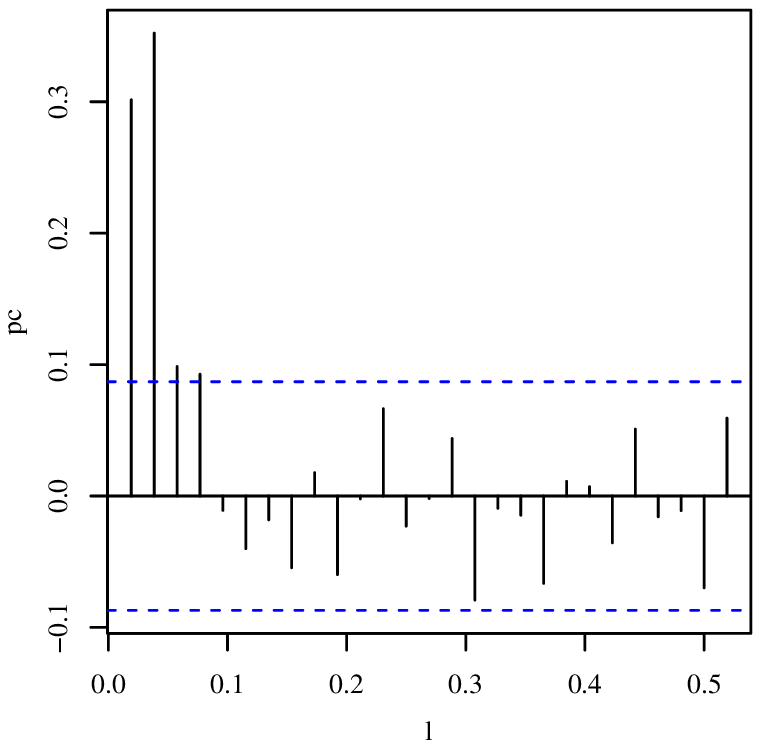}}\\
\vspace{-0.2cm}
 \caption{\small QQ plot and its envelope for the quantile residual, and sample autocorrelation and partial autocorrelation functions of the quantile residuals for 
 the indicated model with the mortality data.}
  \label{fig:QQRes_ex01}
\end{figure}

\subsection{Log-symmetric-ARMAX results}

Now, we present the results based on the proposed log-symmetric-ARMAX model. We also consider $\phi_{t}=\phi$ for $t=1,\ldots,n$. Table~\ref{tab:estimates01} reports the estimates, SEs and $p$-values of the $t$-test for the log-symmetric ARMAX model parameters, as well as the values of AIC, BIC and RMSE. From this table, note that the LogN-ARMAX(2,0) model provides better adjustment compared to the other models based on the values of RMSE, AIC and BIC. Figure~\ref{fig:QQRes_ex02} displays the QQ plots with simulated envelope of the quantile residual for the log-symmetric-ARMAX models. The figure shows that these residuals provide a good agreement with the EXP(1) distribution for the LogN-ARMAX(2,0) model. Note also that all three log-symmetric ARMAX models produce non-autocorrelated residuals according to the sample autocorrelation and partial autocorrelation functions. This result supports the importance of a model which takes into account serial correlation.

\begin{table}[!ht]
\footnotesize
\centering
\renewcommand{\arraystretch}{0.9}
\renewcommand{\tabcolsep}{0.1cm}
\caption{\small {Estimates (with SE in parentheses) and model selection measures for fit to the mortality data.}}\label{tab:estimates02}
\begin{tabular}{llcrcccccccccc}
\hline
Model                &              & Parameter     & \multicolumn{1}{c}{ML estimate}       & $p$-value & RMSE      &   AIC     &   BIC    \\
\hline\\[-0.25cm]
LogN-ARMAX(2,0) model&              & $\kappa_{1}$  &    0.4050(0.0441)                     &           & 0.0547   &$-$1487.895&$-$1454.051  \\
                     &              & $\kappa_{2}$  &    0.2789(0.0452)                     &           &          &           &              \\
                     &              & $\beta_{0}$   & 38.1026(2.1288)                       &$<$0.0001  &          &           &     \\
                     &              & $\beta{1}$    &$-$0.0170(0.0011)                      &$<$0.0001  &          &           &              \\
                     &              & $\beta_{2}$   &$-$0.0017(0.0004)                      &$<$0.0001  &          &           &              \\
                     &              & $\beta{3}$    &   0.0002($<$0.0001)                   &$<$0.0001  &          &           &              \\
                     &              & $\beta_{4}$   &   0.0023(0.0002)                      &$<$0.0001  &          &           &              \\                    
                     &              & $\phi$        &   0.0023(0.0003)                      &           &          &           &              \\\hline\\[-0.25cm]
Log$t$-ARMAX(2,0) model &           & $\kappa_{1}$  &    0.4527(0.0429)                     &           & 0.0550   &$-$1479.249&$-$1445.406  \\
                     &              & $\kappa_{2}$  &    0.2637(0.0415)                     &           &          &           &              \\
                     &              & $\beta_{0}$   & 38.0715(2.0071)                       &$<$0.0001  &          &           &     \\
                     &              & $\beta{1}$    &$-$0.0170(0.0010)                      &$<$0.0001  &          &           &              \\
                     &              & $\beta_{2}$   &$-$0.0015(0.0004)                      &$<$0.0001  &          &           &              \\
                     &              & $\beta{3}$    &   0.0002($<$0.0001)                   &$<$0.0001  &          &           &              \\
                     &              & $\beta_{4}$   &   0.0021(0.0002)                      &$<$0.0001  &          &           &              \\                    
                     &              & $\phi$        &   0.0025(0.0003)                      &           &          &           &              \\
                     &              & $\vartheta$   &   9                                   &           &          &           &              \\\hline\\[-0.25cm]
LogPE-ARMAX(2,0) model&             & $\kappa_{1}$  &    0.3937(0.0452)                     &           & 0.0548   &$-$1486.880&$-$1453.037  \\
                     &              & $\kappa_{2}$  &    0.2917(0.0465)                     &           &          &           &              \\
                     &              & $\beta_{0}$   & 38.1097(2.0483)                       &$<$0.0001  &          &           &     \\
                     &              & $\beta{1}$    &$-$0.0170(0.0010)                      &$<$0.0001  &          &           &              \\
                     &              & $\beta_{2}$   &$-$0.0016(0.0004)                      &$<$0.0001  &          &           &              \\
                     &              & $\beta{3}$    &   0.0002($<$0.0001)                   &$<$0.0001  &          &           &              \\
                     &              & $\beta_{4}$   &   0.0023(0.0002)                      &$<$0.0001  &          &           &              \\                    
                     &              & $\phi$        &   0.0020(0.0002)                      &           &          &           &              \\
                     &              & $\vartheta$   &   0.24                                 &           &          &           &              \\\hline\\[-0.25cm]
\end{tabular}
\end{table}

\begin{figure}[!ht]
\centering

\psfrag{-1}[c][c]{\tiny{$-$1}}
\psfrag{-2}[c][c]{\tiny{$-$2}}
\psfrag{-3}[c][c]{\tiny{$-$3}}
\psfrag{-4}[c][c]{\tiny{$-$4}}
\psfrag{0}[c][c]{\tiny{0}}
\psfrag{1}[c][c]{\tiny{1}}
\psfrag{2}[c][c]{\tiny{2}}
\psfrag{3}[c][c]{\tiny{3}}
\psfrag{4}[c][c]{\tiny{4}}
\psfrag{5}[c][c]{\tiny{5}}
\psfrag{6}[c][c]{\tiny{6}}
\psfrag{8}[c][c]{\tiny{8}}
\psfrag{-}[c][c]{\tiny{$-$}}
\psfrag{2.0}[c][c]{\tiny{2.0}}
\psfrag{2.5}[c][c]{\tiny{2.5}}
\psfrag{-0.1}[c][c]{\tiny{$-$0.1}}
\psfrag{-2.0}[c][c]{\tiny{$-$2.0}}
\psfrag{-2.5}[c][c]{\tiny{$-$2.5}}
\psfrag{0.0}[c][c]{\tiny{0.0}}
\psfrag{0.1}[c][c]{\tiny{0.1}}
\psfrag{0.2}[c][c]{\tiny{0.2}}
\psfrag{0.3}[c][c]{\tiny{0.3}}
\psfrag{0.4}[c][c]{\tiny{0.4}}
\psfrag{0.5}[c][c]{\tiny{0.5}}
\psfrag{0.6}[c][c]{\tiny{0.6}}
\psfrag{0.7}[c][c]{\tiny{0.7}}
\psfrag{0.8}[c][c]{\tiny{0.8}}
\psfrag{0.9}[c][c]{\tiny{0.9}}
\psfrag{1.0}[c][c]{\tiny{1.0}}
\psfrag{10}[c][c]{\tiny{10.0}}
\psfrag{15}[c][c]{\tiny{15.0}}
\psfrag{20}[c][c]{\tiny{20.0}}
\psfrag{l}[c][c]{\tiny{Lag}}
\psfrag{ac}[c][c]{\tiny{Autocorrelation function}}
\psfrag{pc}[c][c]{\tiny{Partial autocorrelation function}}
\psfrag{EQ}[c][c]{\tiny{Empirical quantile}}
\psfrag{TQ}[c][c]{\tiny{Theoretical quantile}}
\subfigure[LogN-ARMAX(2,0)]{\includegraphics[height=3.8cm,width=3.8cm]{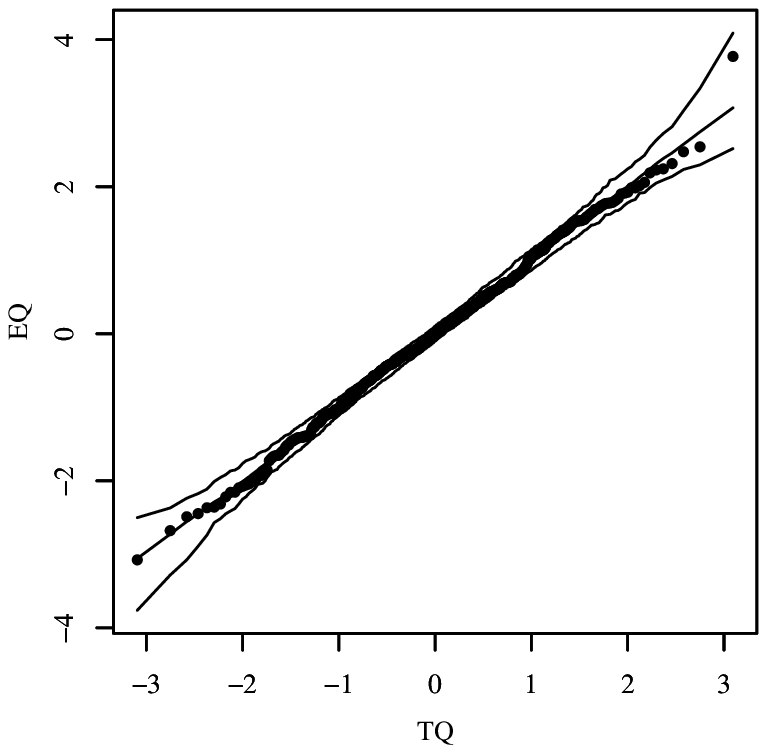}}
\subfigure[LogN-ARMAX(2,0)]{\includegraphics[height=3.8cm,width=3.8cm]{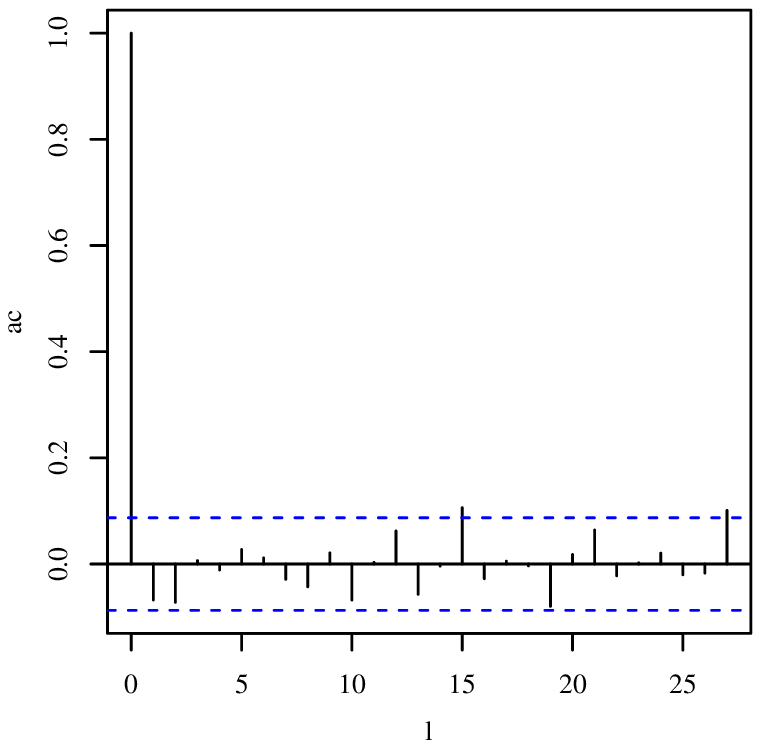}}
\subfigure[LogN-ARMAX(2,0)]{\includegraphics[height=3.8cm,width=3.8cm]{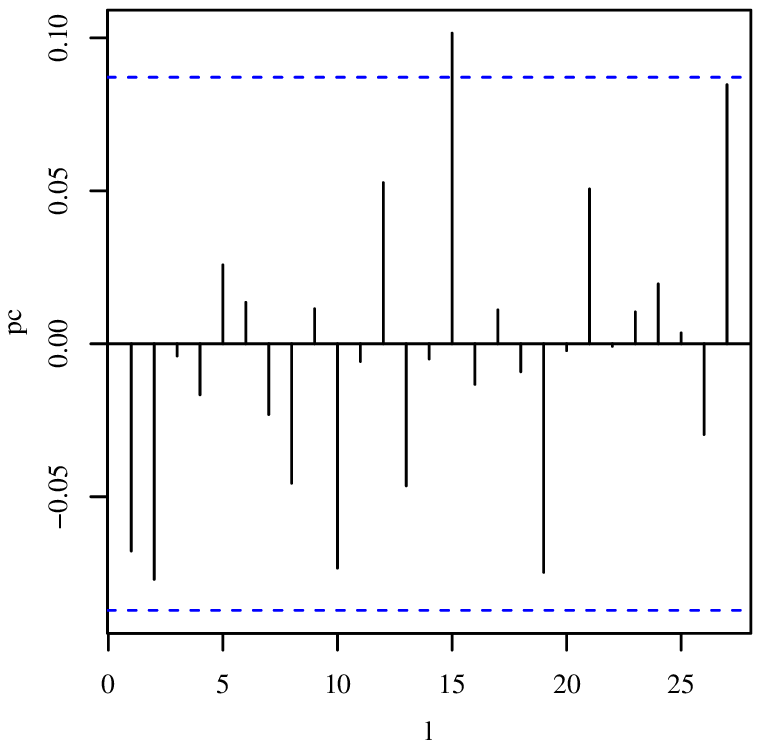}}\\
\subfigure[Log$t$-ARMAX(2,0)]{\includegraphics[height=3.8cm,width=3.8cm]{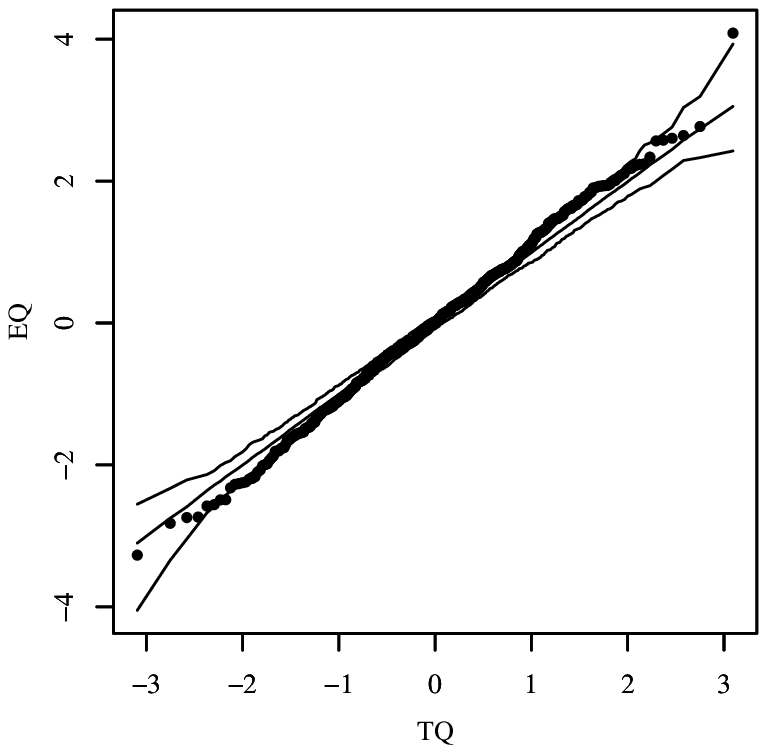}}
\subfigure[Log$t$-ARMAX(2,0)]{\includegraphics[height=3.8cm,width=3.8cm]{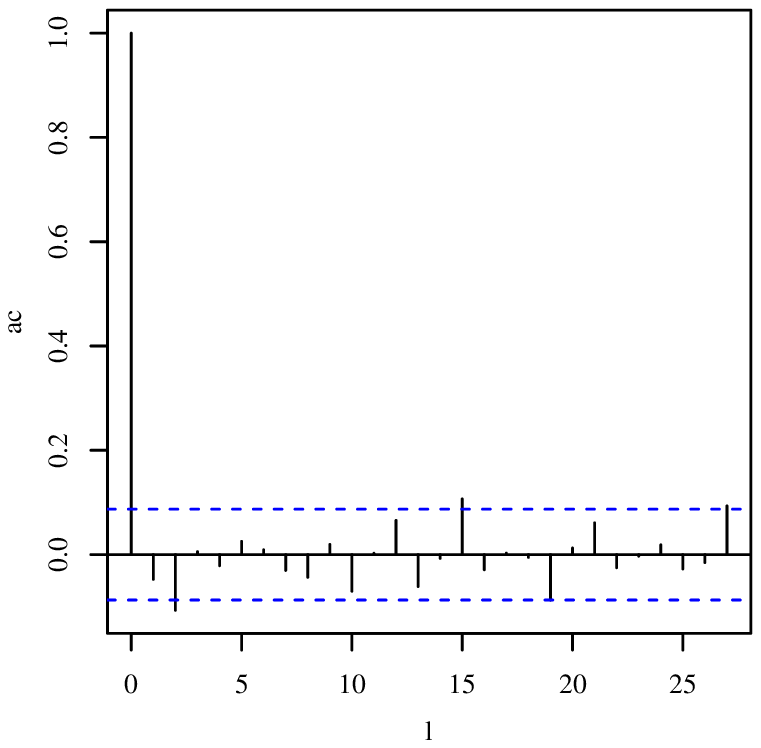}}
\subfigure[Log$t$-ARMAX(2,0)]{\includegraphics[height=3.8cm,width=3.8cm]{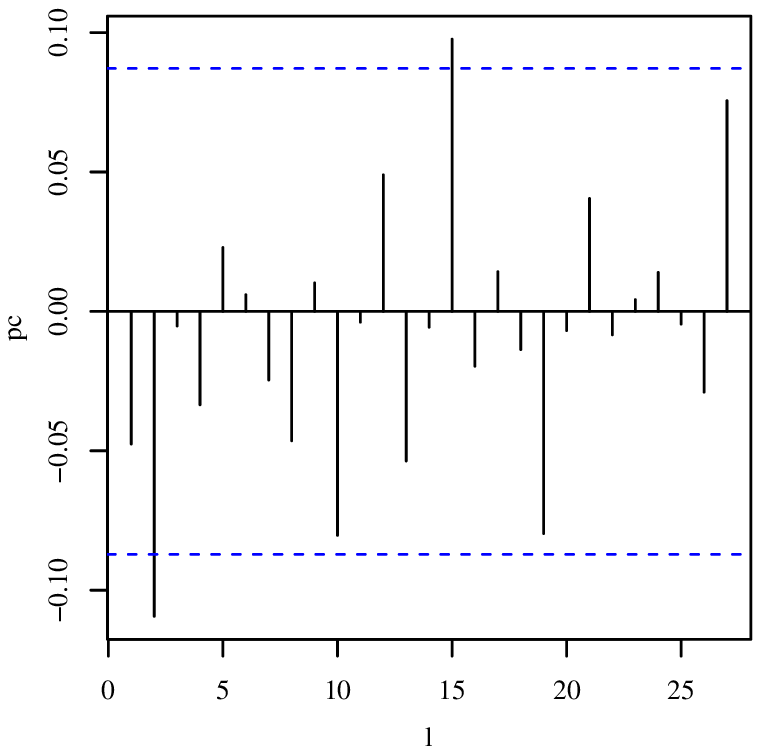}}\\
\subfigure[LogPE-ARMAX(2,0)]{\includegraphics[height=3.8cm,width=3.8cm]{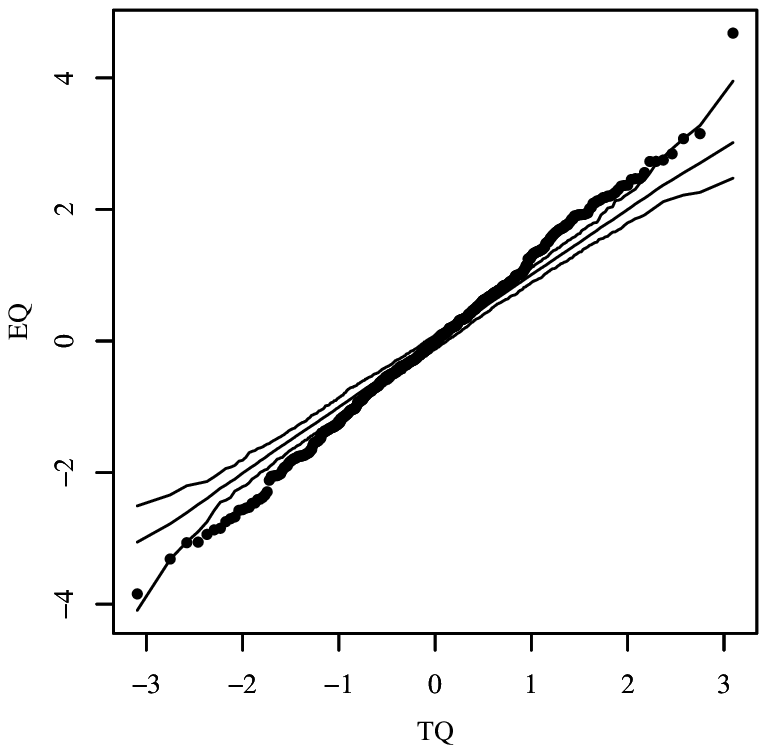}}
\subfigure[LogPE-ARMAX(2,0)]{\includegraphics[height=3.8cm,width=3.8cm]{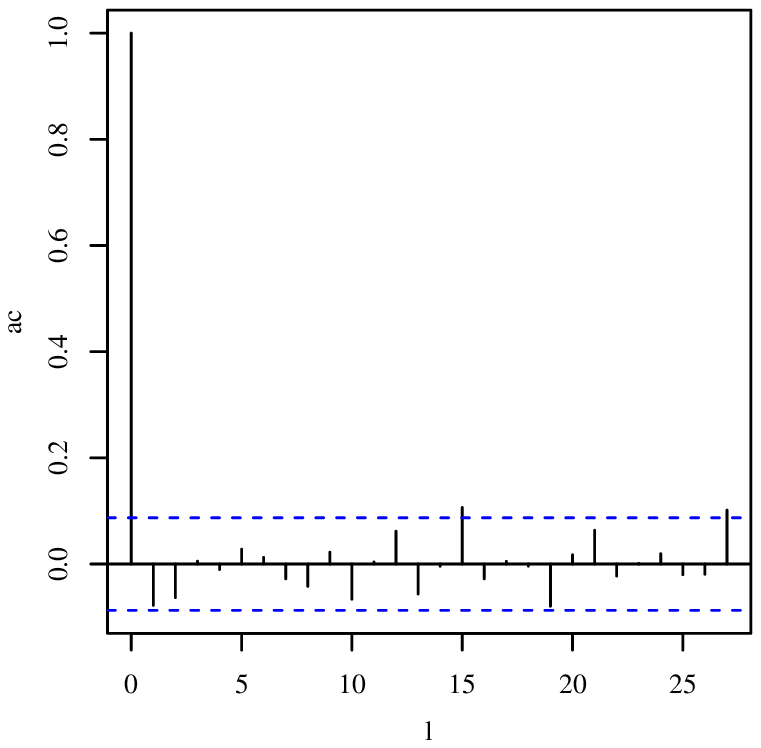}}
\subfigure[LogPE-ARMAX(2,0)]{\includegraphics[height=3.8cm,width=3.8cm]{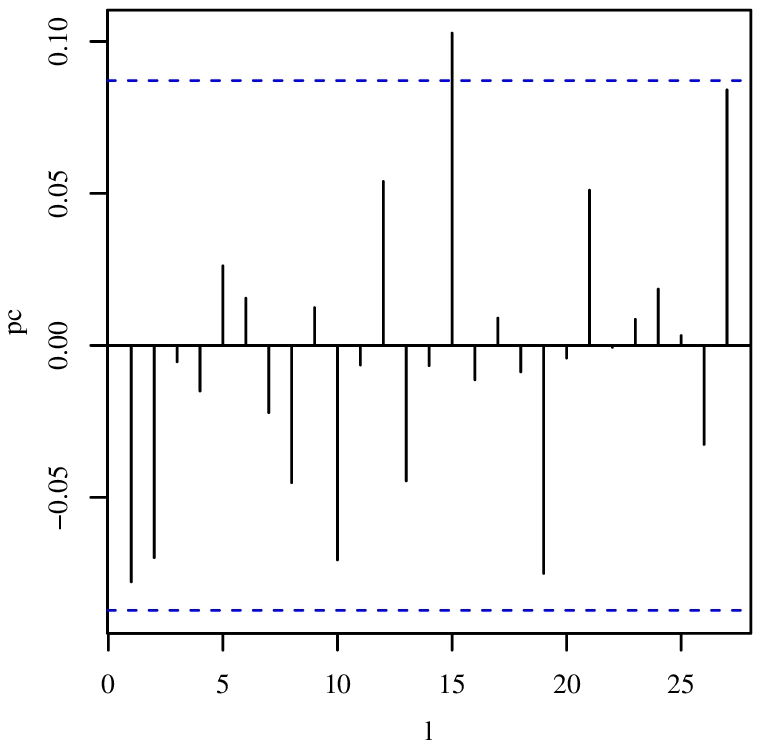}}\\
\vspace{-0.2cm}
 \caption{\small QQ plot and its envelope for the quantile residual, and sample autocorrelation and partial autocorrelation functions of the quantile residuals for 
 the indicated model with the mortality data.}
 \label{fig:QQRes_ex02}
\end{figure}

\section{Concluding remarks}\label{sec:6}
We have proposed a new class of log-symmetric regression models for dealing with cases where the errors are correlated with each other. The proposed approach is an autoregressive and moving average model with covariates and a log-symmetric conditional distribution. We have considered inference about the model parameters and a type of residual for these models. A Monte Carlo simulation study was carried out to evaluate the behavior of the conditional maximum likelihood estimators of the corresponding parameters. We have applied the proposed models to a real-world mortality data set. In general, the results have shown that the proposed models deal with serial correlation quite satisfactory and have great potential in many areas where the modelling of positive and autocorrelated data is necessary. 
As part of future research, it is of interest to discuss influence diagnostic tools and multivariate models. Related ARMA models based on the exponential family and the beta and symmetric distributions can be found in \cite{brs:03}, \cite{rochacribari:09}, \cite{zhengetal:15} and \cite{maiorcysneiros:18}, and multivariate versions of these models can be proposed as well. Work on some of these issues is currently in progress and we hope to report some findings in future papers.


\clearpage
\small

\begin{thebibliography}{}

\bibitem[Benjamin et~al., 2003]{brs:03}
Benjamin, M.~A., Rigby, R.~A., and Stasinopoulos, D.~M. (2003).
\newblock {Generalized autoregressive moving average models}.
\newblock {\em Journal of the American Statistical Association}, 98:214--223.

\bibitem[Crow and Shimizu, 1988]{cs:88}
Crow, E.~L. and Shimizu, K. (1988).
\newblock {\em {Lognormal Distributions: Theory and Applications}}.
\newblock Dekker, New York, US.

\bibitem[Dunn and Smyth, 1996]{ds:96}
Dunn, P. and Smyth, G. (1996).
\newblock {Randomized quantile residuals}.
\newblock {\em Journal of Computational and Graphical Statistics}, 5:236--244.

\bibitem[Efron and Hinkley, 1978]{eh:78}
Efron, B. and Hinkley, D.~V. (1978).
\newblock {Assessing the accuracy of the maximum likelihood estimator:
  {O}bserved vs. expected Fisher information}.
\newblock {\em Biometrika}, 65:457--487.

\bibitem[Fang et~al., 1990]{fkn:90}
Fang, K.~T., Kotz, S., and Ng, K.~W. (1990).
\newblock {\em Symmetric Multivariate and Related Distributions}.
\newblock Chapman and Hall, London, UK.

\bibitem[Maior and Cysneiros, 2018]{maiorcysneiros:18}
Maior, V. Q.~S. and Cysneiros, J.~A. (2018).
\newblock {{SYMARMA}: a new dynamic model for temporal data on conditional
  symmetric distribution}.
\newblock {\em Statistical Papers}, 59:75--97.

\bibitem[Medeiros and Ferrari, 2017]{franciscosilvia2017}
Medeiros, F. M.~C. and Ferrari, S. L.~P. (2017).
\newblock Small-sample testing inference in symmetric and log-symmetric linear
  regression models.
\newblock {\em Statistica Neerlandica}, 71:200--224.

\bibitem[Mittelhammer et~al., 2000]{mjm:00}
Mittelhammer, R.~C., Judge, G.~G., and Miller, D.~J. (2000).
\newblock {\em {Econometric Foundations}}.
\newblock Cambridge University Press, New York, US.

\bibitem[Rocha and Cribari-Neto, 2009]{rochacribari:09}
Rocha, A.~V. and Cribari-Neto, F. (2009).
\newblock {Beta autoregressive moving avarege models}.
\newblock {\em Test}, 18:529--545.

\bibitem[Saulo and Le\~ao, 2017]{saulo2017log}
Saulo, H. and Le\~ao, J. (2017).
\newblock On log-symmetric duration models applied to high frequency financial
  data.
\newblock {\em Economics Bulletin}, 37:1089--1097.

\bibitem[Shumway and Stoffer, 2017]{shusto:17}
Shumway, R.~H. and Stoffer, D.~S. (2017).
\newblock {\em {Time Series Analysis and Its Applications}}.
\newblock Sprinder, Cham, Switzerland.

\bibitem[Vanegas and Paula, 2017]{vanegas2017log}
Vanegas, L. and Paula, G.~A. (2017).
\newblock Log-symmetric regression models under the presence of non-informative
  left-or right-censored observations.
\newblock {\em TEST}, 26:405--428.

\bibitem[Vanegas and Paula, 2016a]{vp:16a}
Vanegas, L.~H. and Paula, G.~A. (2016a).
\newblock An extension of log-symmetric regression models: {R} codes and
  applications.
\newblock {\em Journal of Statistical Simulation and Computation},
  86:1709--1735.

\bibitem[Vanegas and Paula, 2016b]{vanegasp:16a}
Vanegas, L.~H. and Paula, G.~A. (2016b).
\newblock Log-symmetric distributions: statistical properties and parameter
  estimation.
\newblock {\em Brazilian Journal of Probability and Statistics}, 30:196--220.

\bibitem[Vanegas and Paula, 2016c]{vanegasp:16b}
Vanegas, L.~H. and Paula, G.~A. (2016c).
\newblock {\em ssym: Fitting Semi-Parametric log-Symmetric Regression Models}.
\newblock R package version 1.5.7.

\bibitem[Ventura et~al., 2018]{vslm:18}
Ventura, M., Saulo, H., Leiva, V., and Monsueto, S.~E. (2018).
\newblock {Log-symmetric regression models: information criteria and
  application to movie business and industry data}.
\newblock {\em Under review}.

\bibitem[Zheng et~al., 2015]{zhengetal:15}
Zheng, T., Xiao, H., and Chen, R. (2015).
\newblock {Generalized {ARMA} models with martingale difference errors}.
\newblock {\em Journal of Econometrics}, 189:492--506.

\end{thebibliography}

\end{document}